\documentclass[letterpaper,twocolumn,10pt]{article}
\usepackage{usenix}

\usepackage{tikz}
\usepackage{amsmath}

\usepackage{makeidx} 
\usepackage{multirow}
\usepackage{graphicx}
\usepackage{amsmath}
\usepackage{amssymb}
\usepackage{amsfonts}

\usepackage{amsthm}
\usepackage{mathtools}

\newtheorem{theorem}{Theorem}

\newtheorem{definition}{Definition}

\def\BibTeX{{\rm B\kern-.05em{\sc i\kern-.025em b}\kern-.08emT\kern-.1667em\lower.7ex\hbox{E}\kern-.125emX}}
    
\usepackage{nicefrac}
\usepackage{siunitx}
\usepackage{array,framed}
\usepackage{booktabs}
\usepackage{
  color,
  float,
  epsfig,
  wrapfig,
  graphics,
  graphicx,
  subcaption
}
\usepackage{textcomp}
\usepackage{setspace}
\usepackage{latexsym,fancyhdr,url}
\usepackage{enumerate}
\usepackage{graphics}
\usepackage{xparse} 
\usepackage{xspace}
\usepackage{multirow}
\usepackage{csvsimple}
\usepackage{balance}

\usepackage{
  tikz,
  pgfplots,
  pgfplotstable
}
\usepackage{hyperref}

\usetikzlibrary{
  shapes.geometric,
  arrows,
  external,
  pgfplots.groupplots,
  matrix
}

\pgfplotsset{compat=1.9}

\usepackage{mathtools,}

\DeclareMathAlphabet{\mathcal}{OMS}{cmsy}{m}{n}

\DeclareGraphicsExtensions{%
    .png,.PNG,%
    .pdf,.PDF,%
    .jpg,.mps,.jpeg,.jbig2,.jb2,.JPG,.JPEG,.JBIG2,.JB2}

\usepackage{xparse}
\newcommand{\bnm}{\begin{newmath}}
\newcommand{\enm}{\end{newmath}}

\newcommand{\bea}{\begin{eqnarray*}}%
\newcommand{\eea}{\end{eqnarray*}}%

\newcommand{\bne}{\begin{newequation}}
\newcommand{\ene}{\end{newequation}}

\newcommand{\bal}{\begin{newalign}}
\newcommand{\eal}{\end{newalign}}

\newenvironment{newalign}{\begin{align}%
\setlength{\abovedisplayskip}{4pt}%
\setlength{\belowdisplayskip}{4pt}%
\setlength{\abovedisplayshortskip}{6pt}%
\setlength{\belowdisplayshortskip}{6pt} }{\end{align}}

\newenvironment{newmath}{\begin{displaymath}%
\setlength{\abovedisplayskip}{4pt}%
\setlength{\belowdisplayskip}{4pt}%
\setlength{\abovedisplayshortskip}{6pt}%
\setlength{\belowdisplayshortskip}{6pt} }{\end{displaymath}}

\newenvironment{newequation}{\begin{equation}%
\setlength{\abovedisplayskip}{4pt}%
\setlength{\belowdisplayskip}{4pt}%
\setlength{\abovedisplayshortskip}{6pt}%
\setlength{\belowdisplayshortskip}{6pt} }{\end{equation}}

\newcounter{ctr}

\newlength{\saveparindent}
\newlength{\saveparskip}
\newcommand{\E}{{\rm I\kern-.3em E}}

\def\negl{\mathsf{negl}}

\def \part {part}

\def \blackslug{\hbox{\hskip 1pt $\square$ \hskip 1pt}}
\def \qed{\quad\blackslug\lower 8.5pt\null\par}

\newcounter{mynote}[section]
\newcommand{\notecolor}{blue}
\newcommand{\thenote}{\thesection.\arabic{mynote}}
\newcommand{\tnote}[1]{\refstepcounter{mynote}{\bf \textcolor{\notecolor}{$\ll$TomR~\thenote: {\sf #1}$\gg$}}}

\newcommand\ignore[1]{}

\newcounter{rcnote}[section]

\newcounter{mrnote}[section]

\newcounter{fknote}[section]

\newcounter{anote}[section]

\DeclareMathSymbol{\mlq}{\mathord}{operators}{``}
\DeclareMathSymbol{\mrq}{\mathord}{operators}{`'}

\newcommand{\rhf}[2]{R_{f, \gamma}}

\DeclareDocumentCommand{\edist}{o o}{
  \ensuremath{
    \IfNoValueTF{#1}{{d}}{{\sf d}(#1,#2)}
  }
}

\newcommand{\olrk}[1]{\ifx\nursymbol#1\else\!\!\mskip4.5mu plus 0.5mu\left(\mskip0.5mu plus0.5mu #1\mskip1.5mu plus0.5mu \right)\fi}

\NewDocumentCommand{\indseq}{ O{1} O{r} }{{#1}\ldots {#2}}

\usepackage{algorithm}      
\usepackage{algpseudocode}  

\usepackage{makecell}

\usepackage{amsthm}
\usepackage{float}

\usepackage{booktabs}

\usepackage{listings}
\usepackage{color}
\allowdisplaybreaks[4] 
\usepackage{autobreak}
\usepackage{graphicx}
\usepackage{lineno}

\usepackage{enumitem}
\usepackage[most]{tcolorbox}

\newcommand{\GF}[1]{\mathbb{F}_{#1}}

\newcommand{\wt}[1]{w_H(#1)\xspace}

\renewcommand{\vec}[1]{\mathbf{#1}\xspace}
\newcommand{\mat}[1]{\mathbf{#1}\xspace}

\newcommand{\txt}[1]{\mathrm{#1}\xspace}
\usepackage{booktabs}

\usepackage[nointegrals]{wasysym}
\newcommand{\Cfill}{\CIRCLE} 
\newcommand{\copen}{\Circle} 
\newcommand{\rc}{\RIGHTcircle}  

\usepackage{cleveref}
\definecolor{deepred}{RGB}{230, 80, 80}
\definecolor{deepgreen}{RGB}{70, 170, 90}

\newcommand{\partialyes}{ \textcolor{deepgreen}{\ding{51}}\textcolor{deepred}{\rotatebox[origin=c]{-7.2}{\kern-0.7em\ding{55}}}\xspace }

\usepackage{tikz}

\newcommand{\KeyGen}{\mathsf{KeyGen}}
\newcommand{\Encode}{\mathsf{Encode}}
\newcommand{\Decode}{\mathsf{Decode}}
\newcommand{\echannel}[0]{\mathcal{E}}

\newcommand{\Encrypt}{\mathsf{Encrypt}}
\newcommand{\Decrypt}{\mathsf{Decrypt}}

\newcommand{\RMRS}{RM-RS }

\usepackage{tikz}
\usetikzlibrary{arrows.meta, positioning}

\usepackage{pifont}

\newcommand{\mypara}[1]{\noindent{\bf {#1}.} \xspace}

\newcommand{\cecinst}[0]{\textsf{Conc-CEC}\xspace}

\usepackage[most]{tcolorbox}

\newcounter{game}
\renewcommand{\thegame}{\arabic{game}}
\usepackage[normalem]{ulem}

\usepackage{mdframed}

\algnewcommand\algorithmicparam{\textbf{Parameter:}}
\algnewcommand\Param{\item[\algorithmicparam]}

\newcommand{\COMMENT}[1]{\hfill$\triangleright$\textit{#1}}

\usepackage{threeparttable}

\begin{document}

\date{}

\title{\Large \bf {Confidential Error-Correcting Codes}}

\author{
{\rm Tianrui Wang\textsuperscript{1}}\ \ \
{\rm Anyu Wang\textsuperscript{2,3,4}$^\dagger$}\ \ \
{\rm Tianshuo Cong\textsuperscript{5,6}$^\dagger$}\ \ \ \\
{\rm Delong Ran\textsuperscript{1}}\ \ \
{\rm Xiaoyun Wang\textsuperscript{2,3,4,7,8}}
\\
\textsuperscript{1}\textit{Institute for Network Sciences and Cyberspace, Tsinghua University} \ \ \ \\
\textsuperscript{2}\textit{Institute for Advanced Study, Tsinghua University} \ \ \ 
\textsuperscript{3}\textit{Zhongguancun Laboratory, Beijing, China} \ \ \ \\
\textsuperscript{4}\textit{State Key Laboratory of Cryptography and Digital Economy Security, Tsinghua University} \ \ \ \\
\textsuperscript{5}\textit{School of Cryptologic Science and Engineering, Shandong University} \ \ \
\\
\textsuperscript{6}\textit{Shandong Key Laboratory of  Artificial Intelligence Security, Shandong University}
\\
\textsuperscript{7}\textit{National Financial Cryptography Research Center, Beijing, China}
\\
\textsuperscript{8}\textit{Shandong Institute of Blockchain, Shandong, China}
}

\maketitle

\begin{abstract}
We introduce Confidential Error-correcting Codes (CECs), a new cryptographic primitive that enables reliable decoding from perturbed codewords while preserving message confidentiality.
We formalize this primitive through three security properties: confidentiality, robustness, and soundness.
To realize this goal, we develop a general framework for constructing CECs based on a Randomness-Security Separation (RSS) structure, which first encodes a nonce to introduce randomness and then masks the encoded payload with a pseudorandom function.
Based on our RSS framework, we instantiate CECs, called Conc-CEC, using concatenated Reed-Muller and Reed-Solomon codes.
The resulting Conc-CEC admits provable guarantees for all three properties.
Notably, compared to PRCs (CRYPTO'24), another recent cryptographic primitive with error-correcting capability, CECs achieve significantly stronger concrete security guarantees under practical parameter settings.
As a primary application of CECs, we further present CECMark, a watermarking scheme for AI-generated images.
Our experiments on real-world generative image models show that CECMark consistently outperforms state-of-the-art watermarking methods, including Gaussian Shading (CVPR'24), PRCMark (ICLR'25), and T2SMark (NeurIPS'25).
In particular, CECMark achieves higher robustness across all 56 evaluated perturbations and ensures a false positive rate below $10^{-5}$ without introducing side effects.
Taken together, our work lays a novel cryptographic foundation for trustworthy watermarking schemes that, for the first time, unify provable security guarantees with strong practical reliability.
\end{abstract}

\section{Introduction}
\label{sect:introduction}
The growth of AI-Generated Content (AIGC) raises misuse concerns~\cite{ALBUSAIDI2024100630,Lucchi_2024}, making reliable source attribution increasingly important.
Against this backdrop, reliable attribution mechanisms have become essential technical building blocks for effective AI governance. 
For example, the EU AI Act~\cite{euaiact} introduces transparency obligations for AI-generated content, which will become applicable from August 2026 and require certain AI-generated outputs to be marked in a machine-readable and detectable manner.
Among the candidate techniques, watermarking~\cite{sok_zhao_sp25} is especially appealing because it allows generated content itself to carry provenance-related information while preserving content quality.


However, watermarking schemes face multiple security threats, such as unauthorized extraction of embedded information, malicious perturbation of watermarked content, and attribution errors.
Therefore, beyond the common utility requirements (e.g., fidelity), these threats lead to three core security requirements: Confidentiality, Robustness, and Soundness.
Despite their importance,
existing watermarking schemes still exhibit two fundamental
limitations:
\begin{itemize}

\item \textbf{Neglected Confidentiality.} Robustness and soundness have been widely studied in prior watermarking literature~\cite{christ2024pseudorandom}, whereas confidentiality has received considerably less attention despite its growing importance in provenance-oriented watermarking systems.

\item \textbf{Lack of Unified Protection.}
Satisfying all three requirements simultaneously is challenging because of the inherent tension between cryptographic randomness and structured redundancy. 
As a result, SOTA schemes typically optimize only a subset of these properties (see~\Cref{tab:sumprop}), and no practical scheme currently achieves all three simultaneously.

\end{itemize}



\noindent{\bf{Why Confidentiality Matters?}}
The Coalition for Content Provenance and Authenticity (C2PA)~\cite{c2pa} emphasizes embedding multi-bit provenance metadata for AIGC authentication and verification.
Major AI providers, including OpenAI~\cite{openai_c2pa}, Claude~\cite{claude}, 
Azure~\cite{azure_content_credentials}, and Google~\cite{google_synthid},
have adopted mechanisms to embed provenance information into AIGCs.
Specifically, the C2PA specification~\cite{c2pa} requires provenance information describing who created or modified an asset and how it was changed.
Such embedded information may include timestamps, user IDs, digital signatures, or model specifications~\cite{DBLP:conf/iclr/GunnZS25}. 
While these attributes are essential for enabling fine-grained provenance and accountability, they may also contain information that should not be accessible to unauthorized parties.
For example, user IDs may compromise user privacy and ownership anonymity, while model specifications may enable adversaries to identify and exploit vulnerabilities in specific generation models.
This motivates confidentiality as a fundamental requirement for multi-bit watermarking, requiring that an adversary should learn nothing about the embedded payload from the watermarked content.

\noindent{\bf{Why Existing Error-Tolerant Cryptographic Schemes Are Insufficient?}}
A natural question is whether existing cryptographic schemes with error-tolerant capability can already support watermarking.
Unfortunately, they typically do not directly fit this setting.
Code- and lattice-based public-key schemes, such as BIKE~\cite{aragon2022bike},
HQC~\cite{gaborit2025hamming}, and Kyber~\cite{bos2018crystals},
introduce noise primarily for security rather than robustness,
while conventional combinations of symmetric encryption and error-correcting codes
do not jointly provide IND-CPA randomness and perturbation robustness.
A detailed analysis is provided in Appendix~\ref{sec:relatedwork_error}.
A particularly promising cryptographic primitive is Pseudorandom Error-correcting Codes
(PRCs)~\cite{christ2024pseudorandom}, introduced at CRYPTO 2024. 
Nevertheless, existing PRC constructions~\cite{christ2024pseudorandom,DBLP:conf/iclr/GunnZS25,ghentiyala2024new,christ2025improved} remain difficult to instantiate securely under practical parameters.
For example, \cite{wang2025cryptanalysis} shows that part of the secret key of the LDPC-based PRC construction in~\cite{christ2024pseudorandom} can be recovered, thereby compromising the confidentiality of the embedded watermark.
We will present a comprehensive analysis of the limitations of all currently known PRC instances in Appendix~\ref{sect:inst_prc}.

\subsection{Our Contributions}

To bridge these gaps, we introduce \emph{Confidential Error-Correcting Codes (CECs)}, a new cryptographic primitive designed to capture message confidentiality, perturbation robustness, and sound decoding within a unified coding framework.
\Cref{fig:paper_org} summarizes the organization of our paper.


\mypara{Definition of CECs}
CECs are defined to satisfy three properties: 
(i) \textit{Confidentiality} requires that any probabilistic polynomial-time adversary has only a negligible advantage in the IND-CPA game, ensuring that the codewords reveal no information about the embedded plaintext information. 
(ii) \textit{Robustness} corresponds to the error-correcting capability, guaranteeing correct decryption even when the ciphertext is subject to bounded perturbation. 
(iii) \textit{Soundness} requires that a random vector be successfully decoded to a valid message only with negligible probability.
The formal definition of CECs will be provided in~\Cref{sec:def_cec}.

\mypara{Construction of CECs}
To systematically generate the codewords of CECs, we develop a generic construction framework based on a \textit{Randomness-Security Separation (RSS)} structure (see~\Cref{fig:high_level}).
This design creates a clean separation of roles: Confidentiality is achieved solely from the security of Pseudorandom Functions (PRFs), without requiring any cryptographic indistinguishability property from the underlying error-correcting codes.
That said, such a separation allows the error-correcting component to focus on Robustness and Soundness, while security is ensured independently.

\begin{figure}[t]
\centering
\includegraphics[width=\linewidth]{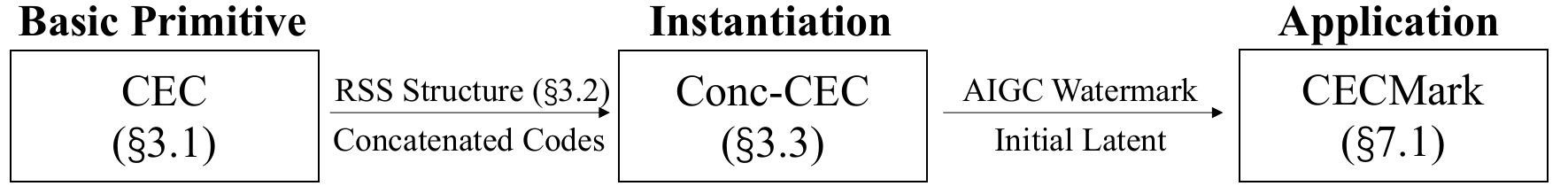}
\caption{Overview of the paper organization. 
}
\label{fig:paper_org}
\end{figure}

\mypara{Instantiation of CECs}
In~\Cref{subsec:conc_cec}, based on the RSS structure, we present \cecinst, which instantiates the error-correcting component using concatenated Reed-Muller and Reed-Solomon (RM and RS) codes. 
This concatenated design provides strong Robustness, supported by a layered Decoding Failure Rate (DFR) analysis in theory.
Meanwhile, we will compare this with duplicated codes in~\Cref{subsec:dfr_dup_rmrs} to highlight the advantage of our design choice.
Moreover, the restricted decoding region reduces the probability of false acceptance, thereby contributing to Soundness.
Above all, \cecinst admits provable security for the following reasons.

\begin{itemize}

\item  \textbf{Provable} \textbf{Confidentiality:} Since the encoded payload is protected only through PRF-based masking, confidentiality directly reduces to the security of the PRF, independent of the error-correcting code. 
\Cref{sect:confidentiality} will provide a detailed theoretical analysis of the provable confidentiality by using the IND-CPA game (Game~\ref{game:indcpa}).


\item \textbf{Provable Robustness:}
Robustness is delegated to the error-correcting component, which ensures reliable recovery under perturbations and noise. In~\Cref{sect:robustness}, we will analyze robustness by deriving the DFR of the concatenated code under a given noise rate. 
The overall DFR is characterized through a layered analysis of the RM and RS components.

\item\textbf{Provable Soundness:}
Soundness is ensured by the bounded decoding radius of the code, which restricts decoding to a small neighborhood of valid codewords, making it unlikely for arbitrary inputs to be mapped to valid messages. The formal proof is given in~\Cref{sect:soundness}.

\end{itemize}

\mypara{CEC-based Watermarking Schemes}
In~\Cref{sect:experiments}, we propose CECMark, a novel watermarking scheme for generative image models based on CECs. 
To the best of our knowledge, CECMark is the first watermarking scheme to unify provable security guarantees with strong practical reliability.
CECMark embeds the codewords of \cecinst into the signs of latent of generative models restricted to those with large absolute values, enabling both security guarantees and reliable detection.
Theoretically, CECMark achieves a $2^{64}$-level randomness and a $2^{256}$-level confidentiality by selecting a $64$-bit nonce and instantiating the underlying PRF with $256$-bit security.
This provides substantially stronger provable security than GS, PRCMark, and T2SMark, whose confidentiality levels are all below $2^{28}$ (\Cref{sec:conf_cecmark}).
For soundness, CECMark maintains a false positive rate below $10^{-5}$ without introducing side effects (\Cref{sec:sound_cecmark}).
We further empirically evaluate CECMark on Stable Diffusion 2.1~\cite{Rombach_2022_CVPR}.
\Cref{sec:filde_cecmark} demonstrates that CECMark preserves high image quality.
For robustness, \Cref{sec:robust_cecmark} shows that CECMark consistently achieves higher detection accuracy than prior watermark schemes across all $56$ evaluated image perturbations.
These results demonstrate the practical applicability of CECs in trustworthy watermarking systems.
More generally, CECs provide a principled foundation for developing cryptographically grounded and reliably deployable watermarking solutions.

\section{Background}
\label{sect:preliminary}
\subsection{Basic Cryptographic Definitions}

We review the cryptographic definitions used by CECs, beginning with PRFs~\cite{bogdanov2017pseudorandom} in Definition~\ref{def:prf}. 


\begin{definition}[Pseudorandom Function]
\label{def:prf}
A pseudorandom function (PRF) is a function $F$ defined over a keyspace $\mathcal{K}$, input space $\mathcal{X}$, and output space $\mathcal{Y}$:
$F: \mathcal{K} \times \mathcal{X} \rightarrow \mathcal{Y}$.
We have $\mathcal{K} = \GF{2}^{K(\lambda)}$, $\mathcal{X} = \GF{2}^{I(\lambda)}$, and $\mathcal{Y} = \GF{2}^{O(\lambda)}$, where $\lambda$ is the security parameter and $K(\cdot), I(\cdot), O(\cdot)$ are polynomial functions of $\lambda$.
The function $F$ is a PRF if it satisfies the following properties:
\begin{enumerate}
    \item \textbf{Efficient Computability:} For any $k \in \mathcal{K}$ and $x \in \mathcal{X}$, the value $F(k,x)$ can be computed in polynomial time.
    \item \textbf{Pseudorandomness:} It is computationally infeasible to distinguish $F(k, \cdot)$ from a truly random function $R: \mathcal{X} \rightarrow \mathcal{Y}$. 
    Formally, for any probabilistic polynomial-time adversary 
$\mathcal{A} = (\mathcal{A}_{\mathsf{GUESS}})$, where 
$\mathcal{A}_{\mathsf{GUESS}}^{\mathcal{O}}$ is given oracle access to $\mathcal{O}$ 
and can make an adaptive polynomial number of queries to $\mathcal{O}$, we require the advantage
\[
\mathrm{Adv}_{\mathcal{A}}^{\mathsf{PRF}}(\lambda) := \left| \Pr[\mathsf{Game}^{\mathsf{PRF}}_{\mathcal{A}}(\lambda)=1]-\frac{1}{2} \right|,
    \]
    is negligible in $\lambda$, where the game is defined in Game~\ref{game:prf}.
\end{enumerate}
\end{definition}

\refstepcounter{game}
\begin{tcolorbox}[
breakable,
enhanced jigsaw,
colback=white,
colframe=black,
coltitle=black,
colbacktitle=white,
title={Game~\thegame~$\mathsf{Game}^{\mathsf{PRF}}_{\mathcal{A}}(\lambda)$},
fonttitle=\bfseries,
boxrule=0.5pt,
boxsep=0.5pt,
top=1pt,
bottom=1pt
]
\label{game:prf}
\begin{enumerate}[leftmargin=1.2em,label=\arabic*.,nosep]
    \item $b \xleftarrow{\$} \{0,1\}$
    \item $k \xleftarrow{\$} \mathcal{K}, \quad 
    R \xleftarrow{\$} \{ \mathcal{X} \to \mathcal{Y} \}$
    
    \item $\mathcal{O}(\cdot) := 
    \begin{cases}
        F(k,\cdot) & \text{if } b = 0, \\
        R(\cdot)   & \text{if } b = 1
    \end{cases}$
    
    \item $b' \leftarrow \mathcal{A}_{\mathsf{GUESS}}^{\mathcal{O}}(1^\lambda)$
    
    \item return $(b' = b)$
\end{enumerate}

\end{tcolorbox}

Proving the existence of PRFs remains a fundamental challenge as it implies $P \ne NP$. Nevertheless, PRFs are widely trusted and can be constructed from one-way functions or Pseudorandom Generators (PRGs), such as the NIST-specified DRBG~\cite{barker2007recommendation} instantiated with AES~\cite{daemen2002design} or HMAC~\cite{krawczyk1997hmac}. Furthermore, stream ciphers such as RC4~\cite{jindal2015rc4}, ISAAC~\cite{shafiee2016isaac}, and ChaCha20~\cite{bernstein2008chacha} can also be converted into PRFs. 
We instantiate the PRF with ChaCha20~\cite{bernstein2008chacha} for its security and broad adoption. 
We next define IND-CPA security in Definition~\ref{def:indcpa}.

\begin{definition}[IND-CPA Security]
\label{def:indcpa}
Let $(\KeyGen, \Encrypt,$ $ \Decrypt)$ be an encryption scheme.
We say it is IND-CPA secure if for any probabilistic polynomial-time adversary 
$\mathcal{A} = (\mathcal{A}_{\mathsf{CHOOSE}}, \mathcal{A}_{\mathsf{GUESS}})$, where 
$\mathcal{A}_{\mathsf{CHOOSE}}^{\mathcal{O}}$ and $\mathcal{A}_{\mathsf{GUESS}}^{\mathcal{O}}$ 
are given oracle access to $\mathcal{O}$ and can make an adaptive polynomial number of queries to $\mathcal{O}$, the advantage
\[
\mathrm{Adv}_{\mathcal{A}}^{\mathsf{IND\text{-}CPA}}(\lambda)
:= \left| \Pr[\mathsf{Game}^{\mathsf{IND\text{-}CPA}}_{\mathcal{A}}(\lambda)=1] - \frac{1}{2} \right|
\]
is negligible in $\lambda$, where the game is defined in Game~\ref{game:indcpa}.
\end{definition}

\refstepcounter{game}
\begin{tcolorbox}[
breakable,
enhanced jigsaw,
colback=white,
colframe=black,
coltitle=black,
colbacktitle=white,
title={Game~\thegame~$\mathsf{Game}^{\mathsf{IND\text{-}CPA}}_{\mathcal{A}}(\lambda)$},
fonttitle=\bfseries,
boxrule=0.5pt,
boxsep=0.5pt,
top=1pt,
bottom=1pt
]
\label{game:indcpa}

\begin{enumerate}[leftmargin=1.2em,label=\arabic*.,nosep]
    \item $\vec{sk} \xleftarrow{\$} \GF{2}^{\lambda_1}$
    
    \item $\mathcal{O}(\cdot) := \Encrypt(1^\lambda,\vec{sk},\cdot)$
    
    \item $(\vec{m}_0,\vec{m}_1) \leftarrow 
    \mathcal{A}_{\mathsf{CHOOSE}}^{\mathcal{O}}(1^\lambda)$
    
    \item $b \xleftarrow{\$} \{0,1\}$
    
    \item $\vec{c}^* \leftarrow \Encrypt(1^\lambda,\vec{sk},\vec{m}_b)$
    
    \item $b' \leftarrow 
    \mathcal{A}_{\mathsf{GUESS}}^{\mathcal{O}}(1^\lambda,\vec{c}^*)$
    
    \item return $(b' = b)$
\end{enumerate}

\end{tcolorbox}

\begin{table*}[t]
\centering
\begin{threeparttable}
\caption{Comparison of SOTA watermarking schemes with respect to the provable security properties. \Cfill~indicates that the property is satisfied, \copen~denotes that it is not achieved, and \rc~signifies that the scheme is provably secure only in theory, but remains insecure in practical deployments.}
\label{tab:sumprop}
\centering
\begin{tabular*}{0.90\textwidth}{@{\extracolsep{\fill}}cccc}
\toprule
\multirow{2}{*}{\textbf{Watermarking Schemes}} &  \multicolumn{3}{c}{\textbf{\textit{Provable} Security Properties}}  \\
\cmidrule(lr){2-4}
& \textbf{Confidentiality} & \textbf{Robustness} & \textbf{Soundness}   \\ \midrule
Gaussian Shading (GS)~\cite{yang2024gaussian}  & \copen & \Cfill & \rc \\
PRCMark~\cite{DBLP:conf/iclr/GunnZS25} & \rc & \Cfill & \Cfill \\ 
T2SMark~\cite{yang2025t2smark} & \rc & \Cfill & \Cfill \\ \midrule
CECMark (\textbf{Ours}) & \Cfill\tnote{(a)} & \Cfill\tnote{(b)} & \Cfill\tnote{(c)} \\
\bottomrule
\end{tabular*}

\begin{tablenotes}
\footnotesize
\item[(a)] Confidentiality is theoretically established but remains practically vulnerable in PRCMark and T2SMark (see \Cref{sect:confidentiality,sec:conf_cecmark}).
\item[(b)] Robustness follows the order CECMark $>$ PRCMark $>$ T2SMark $\approx$ GS (see \Cref{sect:robustness,sec:robust_cecmark}).
\item[(c)] Soundness is inherently ensured in CECMark, whereas PRCMark and T2SMark incur additional side effects (see \Cref{sect:soundness,sec:sound_cecmark}).
\end{tablenotes}

\end{threeparttable}
\end{table*}

\subsection{Coding Theory}
\mypara{Basic Notations}
Let $\GF{q}$ denote a finite field with $q$ elements.
For a vector $\vec{x} \in \GF{q}^n$, let $\wt{\vec{x}}$ denote its Hamming weight.
For two vectors $\vec{x}, \vec{y} \in \GF{q}^n$, the Hamming distance between them is $d(\vec{x}, \vec{y}) = \wt{\vec{x} - \vec{y}}$.
In the following, we introduce some basic concepts related to error-correcting codes.

\begin{definition}[Minimum Distance]
The minimum distance of a code $\mathcal{C}$ is $
d_{\min}(\mathcal{C})=\min \{ d(\vec{x},\vec{y}) : \vec{x},\vec{y}\in \mathcal{C},\ \vec{x}\neq \vec{y}\}$.
\end{definition}

\begin{definition}[Linear Codes]
A $q$-ary linear code $\mathcal{C}$ of length $n$ is a linear subspace of $\GF{q}^n$.
If $\mathcal{C}$ has dimension $k$, then $\mathcal{C}$ is called a $[n,k]$ or $[n,k]_q$ code.
If the minimum distance of $\mathcal{C}$ is $d_{\min}$, we then describe it as a $[n,k,d_{\min}]$ or $[n,k,d_{\min}]_q$ code.
\end{definition}


\begin{definition}[Reed-Muller Codes]
    A Reed-Muller code RM$(r,m)$ with positive integer $m$ and $0\le r \le m$ has parameters
\(
[n,k,d_{\min}] = \bigl[2^m, \sum_{i=0}^{r} \binom{m}{i}, 2^{m-r}\bigr],
\) 
and duplicated RM codes are constructed by repeating each codeword multiple times.
\end{definition}

\begin{definition}[Reed-Solomon Codes]
    Let $p$ be a prime number and $q$ be a power of $p$. 
A Reed-Solomon code
RS[$n$,$k$,$d_{\min}$] with symbols in $\mathbb{F}_q$ has block length $n=q-1$, dimension $k$, and minimum distance $d_{\min}=n-k+1$.
A shortened Reed-Solomon code is obtained from an RS[$n,k,d_{\min}$] code by fixing $s$ message symbols and removing the corresponding $s$ codeword positions. 
The resulting code has parameters RS[$n-s,k-s,d_{\min}$], while preserving the minimum distance.
\end{definition}

Concatenated codes are a powerful class of error-correcting codes that leverage a multi-stage structure to achieve better decoding performance. 
We present the definition of concatenated codes in Definition~\ref{def:concode} and describe their encoding procedure in Appendix~\ref{sect:encode_concatenated}.

\begin{definition}[Concatenated Codes]
\label{def:concode}
A concatenated code over $\GF{2}$ consists of
an external code $\mathcal{C}_e$ with parameters $[n_e,k_e,d_e]$ over $\GF{q}$
and an internal code $\mathcal{C}_i$ with parameters $[n_i,k_i,d_i]$ over $\GF{2}$,
where $q = 2^{k_i}$.
The resulting concatenated code is a binary code with parameters
$[N = n_e n_i,\; K = k_e k_i,\; D \ge d_e d_i]$.
\end{definition}

\subsection{Generative Image Models (GIMs)}
\label{sec:gim}

A Generative Image Model (GIM) is a probabilistic algorithm that takes a textual prompt as input and outputs an image.
We focus on the most widely used latent diffusion-based GIMs, which consist of a de-noising model \texttt{U-net} and an autoencoder pair $(\texttt{E-net}, \texttt{D-net})$.
The generation and inversion processes are illustrated in~\Cref{fig:watermark_flow}.
In the generation process, the model samples an initial latent $\vec{z}_{T}$ from a Gaussian distribution.
The \texttt{U-net} then progressively de-noises this latent into $\vec{z}_{0}$ under the guidance of the textual prompt over $T$ steps.
The decoder \texttt{D-net} finally maps $\vec{z}_{0}$ to the image space and produces the generated image.
In the inversion process, an input image is first encoded by \texttt{E-net} into a proxy latent $\hat{\vec{z}}_{0}$.
The \texttt{U-net} is then used to approximately invert the de-noising trajectory and obtain a proxy initial latent $\hat{\vec{z}}_{T}$.
It allows information embedded in the initial latent $\vec{z}_{T}$ to be approximately recovered from the generated image through generating $\hat{\vec{z}}_{T}$.


\subsection{In-processing Watermarking Schemes}
In-processing watermarking schemes typically leverage the initial latent $\vec{z}_{T}$ to embed watermarks. 
We review the embedding and extraction methods of SOTA schemes, specifically Gaussian Shading~\cite{yang2024gaussian}, PRCMark~\cite{DBLP:conf/iclr/GunnZS25}, and T2SMark~\cite{yang2025t2smark}.
\Cref{tab:sumprop} summarizes their provable security properties and compares them with CECMark.

\mypara{Gaussian Shading (GS)~\cite{yang2024gaussian}}
In the GS scheme, the message is first encoded via duplicated codes and subsequently masked by a vector generated by the secret stream key. 
The resulting vector is then embedded into the signs of the latent.
The extraction process operates inversely: the received vector is first de-masked, and then decoded by the duplicated codes. 
However, due to the lack of randomness, identical messages always produce identical outputs, which fails to provide confidentiality.
Furthermore, the reliance on simple duplicated codes results in limited robustness and soundness.


\mypara{PRCMark~\cite{DBLP:conf/iclr/GunnZS25}} 
PRCMark is a practical instantiation based on the PRC primitive~\cite{christ2024pseudorandom}.
For watermark embedding, it directly adopts the encode algorithm from the LDPC-PRC in~\cite{christ2024pseudorandom} to generate binary codewords, which are then mapped to the signs of the initial latent.
For watermark extraction, it offers a decode function that employs a Belief Propagation (BP) algorithm to recover the message. 
While the scheme requires low LDPC sparsity to maintain robustness, this requirement introduces a vulnerability to polynomial-time key recovery attacks that ultimately compromise confidentiality~\cite{wang2025cryptanalysis}.

\mypara{T2SMark~\cite{yang2025t2smark}}
T2SMark adopts a two-stage design for embedding. 
The first stage of T2SMark involves encrypting an encoded session key with a master key, while the second stage encrypts the encoded message using that session key. 
To ensure robustness, T2SMark utilizes duplicated codes in its encoding process, following the same approach as GS. 
In terms of embedding strategy, while GS and PRCMark directly assign the vector bits to the signs of the latent space, T2SMark introduces an optimization that restricts watermark embedding to positions with high magnitude. 
However, the restricted $16$-bit length of both the session and master keys renders the scheme susceptible to straightforward key recovery. 
Additionally, the use of duplicated codes leads to insufficient robustness and soundness.


\section{Confidential Error-correcting Codes}
\label{sect:specification}
In this section, we propose a new primitive called \emph{Confidential Error-Correcting Codes (CECs)}.  
We first introduce the definition of CECs. 
Next, we propose a construction framework with a randomness-security separation structure.
Finally, we instantiate CECs with concatenated codes.

\subsection{Definition of CECs}
\label{sec:def_cec}
We formalize CECs through three properties: \emph{confidentiality}, \emph{robustness}, and \emph{soundness}.

\begin{definition}[CECs]
\label{def:cec}
Let $\Sigma$ be a fixed alphabet, $\echannel: \Sigma^* \rightarrow \Sigma^*$ be a noisy channel modeled by adding a noise vector with a Bernoulli distribution, and $k(\lambda) \in \mathbb{N}$ be the message length under any fixed $\lambda$. 
A CEC is a triple of polynomial-time randomized algorithms $(\KeyGen, \Encrypt, \Decrypt)$ that satisfies
\begin{itemize}
\item (\textbf{Confidentiality}) For any polynomial-time adversary $\mathcal{A}$, let $\mathrm{Adv}_{\mathcal{A}}^{\mathrm{IND\text{-}CPA}}(\lambda)$ denote the advantage of the IND-CPA game. Then
\[\mathrm{Adv}_{\mathcal{A}}^{\mathrm{IND\text{-}CPA}}(\lambda) \le \negl(\lambda),\]
where the definition of the IND-CPA game is in Game~\ref{game:indcpa}. 
\item (\textbf{Robustness}) For any message $\mathsf{m} \in \Sigma^{k(\lambda)}$, any $\lambda \in \mathbb{N}$, and $\vec{x} \leftarrow \Encrypt(1^\lambda,\mathsf{sk},\mathsf{m})$, 
\[ 
\mathop{\Pr}\limits_{\mathsf{sk} \leftarrow \KeyGen(1^\lambda)}[\mathsf{Decrypt}(1^\lambda,\mathsf{sk},\echannel(\vec{x}))=\mathsf{m}] \ge 1-\negl(\lambda).
\]
\item (\textbf{Soundness}) For any fixed $\vec{c} \in \Sigma^*$,
\[ 
\mathop{\Pr}\limits_{\mathsf{sk} \leftarrow \KeyGen(1^\lambda)}[\Decrypt(1^\lambda, \mathsf{sk}, \vec{c})=\perp ] \ge 1-\negl(\lambda).
\]
\end{itemize}
\end{definition}

Overall, CECs can be viewed as a specialized symmetric-key encryption scheme equipped with error-correcting capability. 
To ensure robustness, this primitive inherently supports perturbations. 
Consequently, unlike traditional symmetric-key encryption schemes, CECs operate over a ciphertext space that is larger than their plaintext space.



\subsection{Construction of CECs}

\mypara{Overview of the RSS Structure}
To construct CECs, we propose a framework featuring a Randomness-Security Separation (RSS) design.
A high-level illustration of the RSS framework is shown in~\Cref{fig:high_level}. 
This architecture streamlines encryption by decoupling it into two pipelines:

\begin{figure}[t]
\centering
\includegraphics[width=0.85\linewidth]{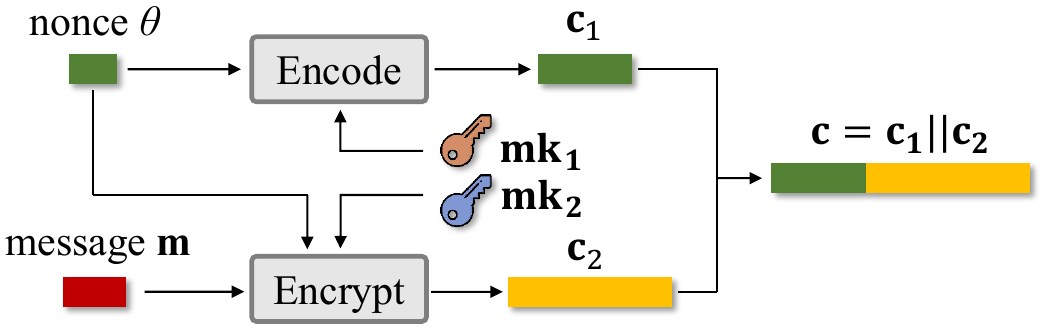}
\caption{A high-level illustration of the Randomness-Security Separation (RSS) design. 
The green pipeline provides randomness, while the yellow pipeline ensures security. 
}
\label{fig:high_level}
\end{figure}

\begin{itemize}
\item \textbf{Randomness Pipeline:} This pipeline processes a random nonce $\theta$ to generate the first ciphertext component $\vec{c}_1$ through an ``\textsf{Encode}'' module, providing the necessary session randomness for the transformation.
Note that although a master key is involved here, it only generates a static mask and does not provide security.

\item \textbf{Security Pipeline:} This branch encrypts the actual message $\vec{m}$ into the second component $\vec{c}_2$ through an ``Encrypt'' module. 
Here, confidentiality is enforced by a mask derived from the master key and the nonce.
This pipeline can meet confidentiality requirements because identical plaintexts yield distinct ciphertexts. 
\end{itemize}

\mypara{Design Benefits of RSS}
RSS offers three advantages: 

\begin{itemize}

\item \textbf{Extensibility:}
Both the \textsf{Encode} and \textsf{Encrypt} modules in RSS can be replaced with different algorithms.
This modularity supports diverse coding schemes, cryptographic primitives, and deployment requirements.
The two modules can also be instantiated independently to meet different performance, security, and application needs.

\item \textbf{Low Overhead: }
To minimize overhead, we utilize a short nonce, while security is ensured by a longer master key. 
This design exploits the gap between adversarial capacity and practical system constraints. 
Although an adversary may possess substantial computational power (e.g., $2^{64}$ operations), the number of encryption queries in real-world watermarking schemes is much smaller. 
By decoupling the randomness space from the security level, we reduce ciphertext overhead without compromising security against realistic adversaries.

\item \textbf{Functional Separation:} By separating nonce-derived randomness in $\vec{c}_1$ from the key-dependent protection in $\vec{c}_2$, the RSS design ensures that the overall security is strictly ensured by the master key.
This separation allows the choice of codes to focus on error-correcting capability without being constrained by security considerations.

\end{itemize}

\begin{figure*}[t]
\centering
\begin{subfigure}{0.48\linewidth}
\centering
\includegraphics[width=0.9\linewidth]{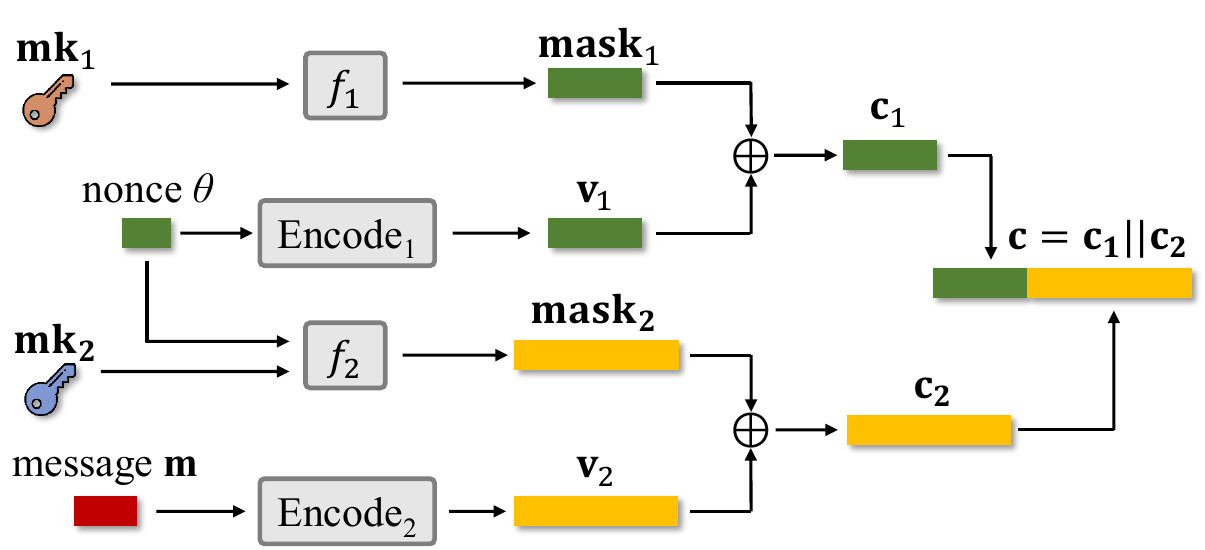}
\caption{Encryption Process.}
\end{subfigure}
\hfill
\begin{subfigure}{0.48\linewidth}
\centering
\includegraphics[width=0.9\linewidth]{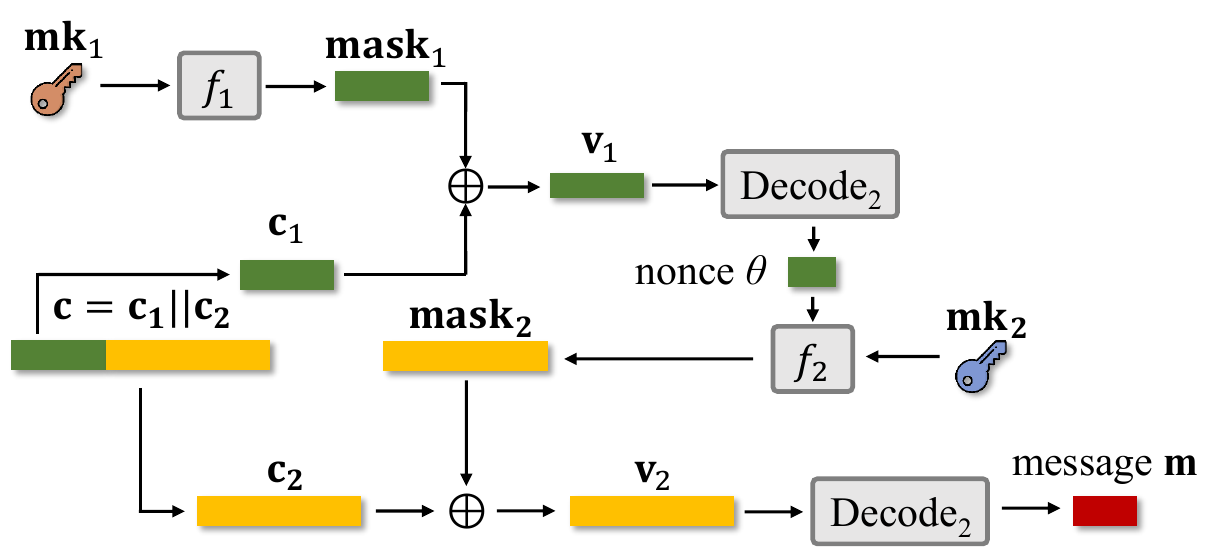}
\caption{Decryption Process.}
\end{subfigure}
\vspace{8pt}
\caption{Encryption and decryption procedures of our RSS-based CEC construction framework.}
\label{fig:cec_encrypt_decrypt}
\end{figure*}

\mypara{Detailed Construction Workflow}
Figure~\ref{fig:cec_encrypt_decrypt} details the workflow of RSS.
Given a security parameter $\lambda$, our construction is parameterized by $\lambda_1,\lambda_2,k,n_1,n_2,n$, where $\lambda_1$ denotes the key length, $\lambda_2$ denotes the nonce length, $k$ is the message length, and the output length is $n=n_1+n_2$.
We utilize two error-correcting codes $\mathcal{C}_1=(\mathsf{Encode}_1, \mathsf{Decode}_1)$ and $\mathcal{C}_2=(\mathsf{Encode}_2, \mathsf{Decode}_2)$.
We also use two functions:
$f_1 : \GF{2}^{\lambda_1} \rightarrow \GF{2}^{n_1}$ and $f_2 : \GF{2}^{\lambda_1} \times \GF{2}^{\lambda_2} \rightarrow \GF{2}^{n_2}$.
Here, $f_1$ takes the first part of the master key $\vec{mk}_1$ as input and expands it into a mask, while $f_2$ derives a mask from the second part $\vec{mk}_2$ together with a nonce.
The formal construction pipeline is described as follows.
\begin{itemize}
    \item $\KeyGen (1^{\lambda})$: 
    Sample $\vec{mk}_1,\vec{mk}_2 \xleftarrow{\$} \GF{2}^{\lambda_1} \times \GF{2}^{\lambda_1}$ and keep the master key $\vec{mk} \leftarrow (\vec{mk}_1,\vec{mk}_2)$  private. 

    \item $\Encrypt (1^{\lambda},\vec{mk},\vec{m})$:
    To encrypt a message $\vec{m}\in \GF{2}^{k}$, 
    first parse $(\vec{mk}_1, \vec{mk}_2) \leftarrow \vec{mk}$.
    Sample a nonce $\vec{\theta} \xleftarrow{\$} \GF{2}^{\lambda_2}$, 
    calculate $\vec{mask}_1 \leftarrow f_1(\vec{mk}_1)$ 
    and 
    $\vec{v}_1 \leftarrow \Encode_1(\vec{\theta})$. 
    Let $\vec{c}_1 \leftarrow \vec{v}_1 \oplus \vec{mask}_1$.
    Calculate $\vec{mask}_2 \leftarrow f_2(\vec{mk}_2,\theta)$.
    Encode the message 
    $\vec{v}_2 \leftarrow \Encode_2(\vec{m})$ and calculate 
    $\vec{c}_2 \leftarrow \vec{v}_2 \oplus \vec{mask}_2$.
    Output $(\vec{c}_1,\vec{c}_2)$.

    \item $\Decrypt(1^\lambda, \vec{mk}, \vec{c})$:
    Let $(\vec{mk}_1,\vec{mk}_2) \leftarrow \vec{mk}$ and $(\vec{c}_1,\vec{c}_2) \leftarrow \vec{c}$.
    Calculate $\vec{mask}_1 \leftarrow f_1(\vec{mk}_1)$.
    Decode the nonce 
    $\vec{\theta}' \leftarrow \txt{Decode}_1(\vec{c}_1 \oplus \vec{mask}_1)$.
    Calculate 
    $\vec{mask}'_2 \leftarrow f_2(\vec{mk}_2,\vec{\theta}')$.
    Decode the message 
    $\vec{m}' \leftarrow \txt{Decode}_2(\vec{c}_2 \oplus \vec{mask}'_2)$ 
    and output $\vec{m}'$.
\end{itemize}

In the randomness pipeline of $\Encrypt(\cdot)$, a uniform random nonce $\theta$ is first encoded into $\vec{v}_1$ and then masked into $\vec{c}_1$.
Although $\vec{c}_1$ includes a mask vector $\vec{mask}_1$, it provides static padding but does not ensure security.
To save storage space, the scheme stores only the seed $\vec{mk}_1$ and derives $\vec{mask}_1$ deterministically from it, which is equivalent to storing $\vec{mask}_1$ directly.
Due to the linearity of the error-correcting code $\mathcal{C}_1$, the codeword $\vec{c}_1$ may leak information about the encoded nonce $\theta$.
For example, let $\vec{c}_1$ and $\vec{c}_1'$ be two codewords generated with the same master key.
An adversary can compute
\[
\begin{aligned}
\Decode_1(\vec{c}_1 + \vec{c}_1')
&= \Decode_1(\vec{v}_1 + \vec{v}_1') \\
&= \Decode_1\bigl(\Encode_1(\theta) + \Encode_1(\theta')\bigr) \\
&= \Decode_1\bigl(\Encode_1(\theta + \theta')\bigr) \\
&= \theta + \theta',
\end{aligned}
\]
which reveals the relation between these two encoded nonces.
This demonstrates that directly encoding sensitive information in this pipeline would violate confidentiality.
To avoid such leakage, we only encode a uniform random nonce $\theta$ in the randomness pipeline.
In the security pipeline of $\Encrypt(\cdot)$, the message $\vec{m}$ is first encoded as $\Encode_2(\vec{m})$ and then masked by a vector
$\vec{mask}_2 \leftarrow f_2(\vec{mk}_2,\theta)$.
The confidentiality relies on the hidden master key $\vec{mk}_2$, while the nonce $\theta$ only provides randomness.
As a result, the encoded message $\vec{m}$ is protected under the mask.

The workflow of $\Decrypt(\cdot)$ follows the reverse procedure of $\Encrypt(\cdot)$: the nonce $\theta'$ is first decoded, which is then used to decrypt the message $\vec{m}'$. 
If the perturbations are within the error-correcting capability of $\mathcal{C}_1$ and $\mathcal{C}_2$, the decryption will be successful.
Notably, in addition to the decoding function, our framework inherently provides a detection function: if either $\mathsf{Decode}_1$ or $\mathsf{Decode}_2$ fails, the detection result is $\mathsf{false}$.


\subsection{\cecinst: An Instantiation of CECs}
\label{subsec:conc_cec}
Based on the proposed RSS framework, we instantiate its components with specific cryptographic and coding algorithms (e.g., concatenated codes). 
We call the resulting \underline{Conc}atenated code-based \underline{CEC} instantiation \cecinst.

Specifically, \cecinst introduces concatenated duplicated Reed-Muller and shortened Reed-Solomon (RM-RS) codes for $\mathcal{C}_1,\mathcal{C}_2$ (i.e., to instantiate $\Encode_1$, $\Decode_1$, $\Encode_2$, $\Decode_2$ in~\Cref{fig:cec_encrypt_decrypt}).
The internal code is obtained by duplicating the RM$(1,7)$ code with parameters $[n_i,k_i,d_i]=[128,8,64]$.
The external code is constructed by shortening an RS code over $\GF{2^8}$.
With RM-RS codes, we adopt a natural failure-handling strategy, where the decryption algorithm outputs a failure symbol whenever either $\mathsf{Decode}_1$ or $\mathsf{Decode}_2$ fails. 
This approach neither affects robustness nor incurs additional overhead.

The motivation for choosing concatenated codes follows the same design principle as HQC~\cite{gaborit2025hamming}, a NIST standard for key encapsulation mechanisms, where the security of the construction is independent of the specific choice of error-correcting codes. Consequently, the coding component is solely responsible for optimizing decoding performance.
In particular, the target application scenarios of CECs lie in the regime with block length $n \approx 10^4$, consistent with HQC, where RM-RS has been shown to perform well.
Therefore, under this principle, RM-RS concatenation is adopted in \cecinst due to its strong empirical performance.

Meanwhile, we instantiate the PRF ($f_1$ and $f_2$ in~\Cref{fig:cec_encrypt_decrypt}) using ChaCha20~\cite{bernstein2008chacha}, a widely deployed and secure stream cipher.
This selection is guided by both security and efficiency considerations. 

In the following three sections, we will conduct a theoretical analysis of the security properties of \cecinst.
It should be emphasized that although the choice of each component is not unique, all subsequent theoretical analyses are conducted with respect to \cecinst.

\section{Confidentiality Analysis of \cecinst}
\label{sect:confidentiality}


The confidentiality of \cecinst is defined by requiring that any adversary achieves only a negligible advantage in the IND-CPA game.
In this section, we prove that the confidentiality of \cecinst holds when secure pseudorandom functions are employed, regardless of the specific error-correcting codes used.

\begin{theorem}[Confidentiality of \cecinst]
\label{thm:ind_cpa}

Let $f_1$ and $f_2$ be independent secure pseudorandom functions.
Then \cecinst satisfies the confidentiality notion in Definition~\ref{def:cec}.
In particular, for any probabilistic polynomial-time adversary $\mathcal{A}=(\mathcal{A}_{\mathsf{CHOOSE}},\mathcal{A}_{\mathsf{GUESS}})$, where $\mathcal{A}$ makes at most $q=\mathsf{poly}(\lambda)$ queries to the oracle $\mathcal{O}$, $Adv_{\mathcal{A}}^{\mathsf{IND\text{-}CPA}}(\lambda)$ is negligible in $\lambda$.
\begin{proof}
The proof of \Cref{thm:ind_cpa} proceeds via a sequence of hybrid games.
We gradually replace the PRFs used in the encryption algorithm with truly random functions.
These replacements have a negligible effect on the advantage because of the security of PRFs.

\vspace{3pt}
\noindent \underline{\textit{Hybrid Games:}}
The proof includes the following games: 
\begin{itemize}
    \item Game~\ref{game:cec-game0} is the IND-CPA game for original \cecinst.
    \item Game~\ref{game:prf-b1} is the PRF game for $f_1$. 
    \item Game~\ref{game:prf-b2} is the PRF game for $f_2$.
    \item Game~\ref{game:cec-game1} is the IND-CPA game for modified \cecinst, where $f_1$ is replaced by a random function.
    \item Game~\ref{game:cec-game2} is the IND-CPA game for modified \cecinst, where $f_1,f_2$ are replaced by random functions.
\end{itemize}

The detailed descriptions of the modified games (Game~\ref{game:cec-game1} and Game~\ref{game:cec-game2}) are deferred to Appendix~\ref{sect:ind_cpa_games}.

\vspace{3pt}
\noindent \underline{\textit{Real Game:}}
Game~\ref{game:cec-game0} corresponds to the real IND-CPA experiment.
For simplicity, we adopt the convention of representing logical outcomes as binary values: any successful decision or match, i.e., $b = b'$ is denoted by 1, while a mismatch is denoted by 0.

\refstepcounter{game}
\begin{tcolorbox}[
breakable,
enhanced jigsaw,
colback=white,
colframe=black,
coltitle=black,
colbacktitle=white,
title={Game~\thegame~$\mathsf{Game}_{\mathcal{A}}^{\mathsf{IND\text{-}CPA}}(\lambda)$ (Original IND-CPA Game)},
fonttitle=\bfseries,
boxrule=0.5pt,
boxsep=0pt,
top=0.5pt,
bottom=0.5pt,
left skip=11pt,
right skip=11pt,
]
\label{game:cec-game0}

\begin{enumerate}[leftmargin=1.2em,label=\arabic*.,nosep]
    \item $(\vec{mk}_1,\vec{mk}_2) \xleftarrow{\$} \GF{2}^{\lambda_1} \times \GF{2}^{\lambda_1}$
    
    \item $\mathcal{O}(\vec{m})$:
    \begin{enumerate}
        \item $\vec{\theta} \xleftarrow{\$} \GF{2}^{\lambda_2}$
        \item $\vec{c}_1 \leftarrow \Encode_1(\vec{\theta}) + f_1(\vec{mk}_1)$
        \item $\vec{c}_2 \leftarrow \Encode_2(\vec{m}) + f_2(\vec{mk}_2,\vec{\theta})$
        \item return $(\vec{c}_1,\vec{c}_2)$
    \end{enumerate}
    
    \item $(\vec{m}_0,\vec{m}_1) \leftarrow 
    \mathcal{A}_{\mathsf{CHOOSE}}^{\mathcal{O}}(1^\lambda)$
    
    \item $b \xleftarrow{\$} \{0,1\}$
    
    \item $\vec{\theta}^* \xleftarrow{\$} \GF{2}^{\lambda_2}$
    
    \item $\vec{c}_1^* \leftarrow \Encode_1(\vec{\theta}^*) + f_1(\vec{mk}_1)$
    
    \item $\vec{c}_2^* \leftarrow \Encode_2(\vec{m}_b) + f_2(\vec{mk}_2,\vec{\theta}^*)$
    
    \item $b' \leftarrow 
    \mathcal{A}_{\mathsf{GUESS}}^{\mathcal{O}}(1^\lambda,(\vec{c}_1^*,\vec{c}_2^*))$
    
    \item return $(b' = b)$
\end{enumerate}

\end{tcolorbox}








\noindent\underline{\textit{Replace $f_1(\vec{mk}_1)$:}}
We next modify Game~\ref{game:cec-game0} by replacing the value $f_1(\vec{mk}_1)$ with a uniformly random vector.
Specifically, in Game~\ref{game:cec-game1}, the value $f_1(\vec{mk}_1)$ is replaced by $\vec{r}_1 \xleftarrow{\$} \GF{2}^{n_1}$.
Equivalently, we replace $f_1(\vec{mk}_1)$ with the output of a truly random function $R_1: \GF{2}^0 \to \GF{2}^{n_1}$.

\vspace{3pt}
\noindent\underline{\textit{PRF Reduction for $f_1$:}}
To justify this replacement between Game~\ref{game:cec-game0} and Game~\ref{game:cec-game1}, we construct a distinguisher $\mathcal{B}_1$ against the pseudorandomness of $f_1$.
The distinguisher $\mathcal{B}_1$ is given oracle access to $\mathcal{O}_1(\cdot)$,
which returns either $f_1(\vec{mk}_1)$ or the output of a truly random function $R_1$.
Instead of directly analyzing the oracle $\mathcal{O}_1$ itself, $\mathcal{B}_1$ simulates the IND-CPA challenger for $\mathcal{A}$ as in Game~\ref{game:prf-b1}.
When $\mathcal{A}$ wins its game, i.e., $\hat{\beta}=\beta$, $\mathcal{B}_1$ guesses $b'=0$, i.e., $\mathcal{O}_1=f_1(\vec{mk}_1)$.

If $\mathcal{O}_1=f_1(\vec{mk}_1)$, the simulation is same as Game~\ref{game:cec-game0}. Let
\[
p_0 := \Pr[\mathsf{Game}~\ref{game:cec-game0}=1] =\Pr[\mathcal{A} \text{ wins Game~\ref{game:cec-game0}}] .
\]

If $\mathcal{O}_1=R_1$, the simulation is same as Game~\ref{game:cec-game1}. Let
\[
p_1 := \Pr[\mathsf{Game}~\ref{game:cec-game1}=1] =\Pr[\mathcal{A} \text{ wins Game~\ref{game:cec-game1}}] .
\]
By the definition of PRF advantage, $\mathrm{Adv}^{\mathrm{PRF}}_{\mathcal{B}_1}$ can be written as
\[
\begin{aligned}
&\quad \  |\frac{1}{2} \big( \Pr[\mathcal{B}_1 \text{ wins } \mid \mathcal{O}_1=f_1] 
+ \Pr[\mathcal{B}_1 \text{ wins } \mid \mathcal{O}_1={R_1}] \big) - \frac{1}{2} | \\
&= |\frac{1}{2} \big( \Pr[\mathcal{A} \text{ wins Game~\ref{game:cec-game0} }] 
+ \Pr[\mathcal{A} \text{ loses Game~\ref{game:cec-game1} }] \big) - \frac{1}{2} | \\
&= |\frac{p_0 + 1 - p_1}{2} - \frac{1}{2} | = \frac{|p_0 - p_1|}{2}.
\end{aligned}
\]
Hence, the difference between Game~\ref{game:cec-game0} and Game~\ref{game:cec-game1} is bounded by
\begin{equation}
\label{eq:indgame_01}
\bigl|\Pr[\mathsf{Game}~\ref{game:cec-game0}=1] - \Pr[\mathsf{Game}~\ref{game:cec-game1}=1]\bigr| \le 2 \,\mathrm{Adv}^{\mathrm{PRF}}_{\mathcal{B}_1}(\lambda).
\end{equation}


\noindent\underline{\textit{Replace $f_2(\vec{mk}_2,\theta)$:}}
We next replace $f_2(\vec{mk}_2,\cdot)$ with a truly random function in Game~\ref{game:cec-game2}.
That is, $f_2(\vec{mk}_2,\cdot)$ is replaced by $R_2 : \GF{2}^{\lambda_2} \to \GF{2}^{n_2}$.
To ensure consistency, we maintain a table $\mathcal{D}$.
For each query on input $\theta$, if $\theta \notin \mathcal{D}$, we sample 
$\mathcal{D}[\theta] \xleftarrow{\$} \GF{2}^{n_2}$.
The output corresponding to $\theta$ is then defined as $\mathcal{D}[\theta]$.

\vspace{3pt}
\noindent\underline{\textit{PRF Reduction for $f_2$:}}
To justify this replacement between Game~\ref{game:cec-game1} and Game~\ref{game:cec-game2}, we construct a distinguisher $\mathcal{B}_2$ against the pseudorandomness of $f_2$.
The distinguisher $\mathcal{B}_2$ is given oracle access to $\mathcal{O}_2(\cdot)$,
which is either the pseudorandom function $f_2(\vec{mk}_2,\cdot)$ with a secret $\vec{mk}_2$ 
or a truly random function $R_2(\cdot)$.
$\mathcal{B}_2$ simulates the IND-CPA challenger for $\mathcal{A}$ as in Game~\ref{game:prf-b2}.
If $\mathcal{O}_2 = f_2(\vec{mk}_2,\cdot)$, then the simulation is identical to Game~\ref{game:cec-game1}.
If $\mathcal{O}_2 = R_2(\cdot)$, then the simulation is identical to Game~\ref{game:cec-game2}.

Let $p_1 := \Pr[\mathsf{Game}~\ref{game:cec-game1}=1]$ and $p_2 := \Pr[\mathsf{Game}~\ref{game:cec-game2}=1]$.
Using an analysis similar to that in the PRF reduction for $f_1$, we have
\[
\mathrm{Adv}^{\mathrm{PRF}}_{\mathcal{B}_2}(\lambda)
= \frac{|p_1 - p_2|}{2}.
\]
Hence,
\begin{equation}
    \label{eq:indcpa_12}
\bigl|\Pr[\mathsf{Game}~\ref{game:cec-game1}=1]
- \Pr[\mathsf{Game}~\ref{game:cec-game2}=1]\bigr|
\le 2\mathrm{Adv}^{\mathrm{PRF}}_{\mathcal{B}_2}(\lambda).
\end{equation}

\noindent\underline{\textit{Analysis of Game~\ref{game:cec-game2}:}}
In Game~\ref{game:cec-game2}, $\vec{c}_1$ is masked by a uniform random string, 
and $\vec{c}_2$ is masked by the output of a random function $R_2$.
\[
\vec{c}_1 = \txt{Encode}_1(\theta) + r_1, \qquad
\vec{c}_2 = \txt{Encode}_2(\vec{m}_b) + R_2(\theta).
\]

Let $\mathsf{Bad}$ denote the event that the nonce $\theta^*$ is used in any of the $q$ queries to the encryption oracle.
We have $\Pr[\mathsf{Bad}] \le \frac{q}{2^{\lambda_2}}$.
Conditioned on $\neg \mathsf{Bad}$, the value $R_2(\theta^*)$ has never been queried before,
and hence is a fresh uniform random string independent of the adversary's view.
Moreover, $r_1$ is also a uniform random string independent of everything else.
Therefore, the challenge ciphertext $(\vec{c}_1^*, \vec{c}_2^*)$ is distributed independently
of the challenge bit $b$.
It follows that $\Pr[b' = b \mid \neg \mathsf{Bad}] = \tfrac{1}{2}.$
Consequently, $\Pr[b' = b]$ can be bounded as
\begin{equation}
    \label{eq:nonce_bad}
    \begin{aligned}
&\Pr[b' = b \mid \neg \mathsf{Bad}]\Pr[\neg \mathsf{Bad}]
 + \Pr[b' = b \mid \mathsf{Bad}]\Pr[\mathsf{Bad}] \\
&\le \frac{1}{2}\bigl(1-\Pr[\mathsf{Bad}]\bigr) + \Pr[\mathsf{Bad}] \le \frac{1}{2} + \frac{q}{2^{\lambda_2+1}} .
\end{aligned}
\end{equation}
Combining the results in Eq.~\eqref{eq:indgame_01}, Eq.~\eqref{eq:indcpa_12}, and Eq.~\eqref{eq:nonce_bad}, we have
\begin{equation}
    \label{eq:adv_indcpa}
    \mathrm{Adv}_{\mathcal{A}}^{\mathrm{IND\text{-}CPA}}(\lambda)
\le 2\,\mathrm{Adv}_{\mathcal{B}_1}^{\mathrm{PRF}}(\lambda)
+ 2\,\mathrm{Adv}_{\mathcal{B}_2}^{\mathrm{PRF}}(\lambda)
+ \frac{q}{2^{\lambda_2+1}}.
\end{equation}

Since $q=\mathsf{poly}(\lambda)$, $\lambda_2=\mathsf{poly}(\lambda)$, and both $f_1,f_2$ are secure PRFs,
the advantage $\mathrm{Adv}_{\mathcal{A}}^{\mathrm{IND\text{-}CPA}}(\lambda)$ is negligible in $\lambda$, which completes the proof.
\end{proof}
\end{theorem}

\section{Robustness Analysis of \cecinst}
\label{sect:robustness}
The robustness of \cecinst can be characterized through the Decoding Failure Rate (DFR) of the underlying concatenated codes $\mathcal{C}_1$ and
$\mathcal{C}_2$.
The DFR of RM-RS codes is given in Theorem~\ref{thm:rs_code}, based on the analysis of~\cite{gaborit2025hamming}.

\begin{theorem}[DFR of the concatenated code~\cite{gaborit2025hamming}] 
\label{thm:rs_code}
The DFR of the concatenated code using an RS code $[n_e,k_e,d_e]_{\GF{256}}$ as the external code and a duplicated RM code as the internal code can be upper bounded as follows:
\begin{equation}
\label{eq:rs_dfr}
\mathrm{DFR}=\sum_{l=\delta_e+1}^{n_e} \binom{n_e}{l}p_i^l(1-p_i)^{n_e-l},
\end{equation}
where $d_e=2\delta_e+1$ and $p_i$ is the upper bound of the DFR of RM codes, as defined in Theorem~\ref{thm:rm_code}.
    
\end{theorem}

\begin{theorem}[DFR of the internal code~\cite{gaborit2025hamming}]
\label{thm:rm_code}
Let $p$ denote the noise rate of the channel. The DFR 
of a duplicated RM code with dimension $k=8$ and minimum distance $d_i$ 
is upper-bounded by $p_i$, which is computed as
\[
\mathsf{min} \Big\{ \sum_{w=d_i/2}^n \mathcal{A}_w p^w (1-p)^{n-w}, 255\sum_{j=d_i/2}^{d_i}\binom{d_i}{j}p^j(1-p)^{d_i-j} \Big\},
\]
where
\[
\begin{aligned}
\mathcal{A}_w
= \mathsf{min}\Big\{\binom{n}{w},\;
&\frac{1}{2}255\binom{d_i}{d_i/2}\binom{d_i}{w-d_i/2} \\
& + 255 \sum_{j=d_i/2+1}^{d_i} \binom{d_j}{j}\binom{d_i}{w-j} \\
&+ \frac{1}{2}\binom{255}{2}\sum_{j=0}^{d_i/2}
\binom{d_i/2}{j}^3\binom{d_i/2}{w-d_i+j}
\Big\}.
\end{aligned}
\]
\end{theorem}



In~\Cref{subsec:dfr_dup_rmrs}, we will provide a concrete DFR analysis of RM-RS
codes under practical parameter settings.
We will further compare RM-RS codes with alternative constructions, including duplicated codes used in GS and T2SMark.
This analysis justifies our choice of RM-RS codes by showing, both theoretically and empirically, that they provide superior robustness through a lower DFR.

\section{Soundness Analysis of \cecinst}
\label{sect:soundness}

We define the False Positive Rate (FPR) as the probability that a random vector is successfully decrypted.
Recall that in~\Cref{subsec:conc_cec}, \cecinst outputs a failure symbol if either $\Decode_1$ or $\Decode_2$ fails.
Therefore, a false positive occurs only when both decoding procedures succeed on a random input.
For the underlying \RMRS code, we define the FPR analogously as the probability that a random vector is successfully decoded. 
We begin by analyzing the FPR of a single \RMRS code.


\begin{theorem}[FPR of the RM-RS Code]
Let $\mathcal{C}\subseteq \GF{q}^{n_e}$ be an RS code RS$[n_e,k_e]$ over $\GF{q}$,
with minimum distance $d_e = n_e-k_e+1$ and error-correcting radius
$
t_e = \left\lfloor \frac{d_e-1}{2} \right\rfloor.
$
Let $\vec{x}\in\GF{2}^{n}$ be uniformly random.
Then the probability that $\vec{x}$ is decoded to some codeword in $\mathcal{C}$
under bounded-distance decoding equals
\begin{equation}
    \label{equ:fpr}
    \mathrm{FPR}
= 
\sum_{i=0}^{t_e} \binom{n_e}{i}  (q-1)^i  q^{k_e-n_e}.
\end{equation}

\begin{proof}
Consider the first-stage decoding of RM.
The input $\vec{x}$ is divided into $n_e$ equal blocks
$\vec{x} = (\vec{x}^{(1)} \| \cdots \| \vec{x}^{(n_e)})$.
Let
\[
\vec{y} = (\vec{y}^{(1)} \| \cdots \| \vec{y}^{(n_e)}),
\quad
\vec{y}^{(i)} = \mathsf{RM.Decode}(\vec{x}^{(i)}),
\]
be the output of the RM decoder.

Since $\vec{x}$ is uniformly random over $\GF{2}^n$ and
$\mathsf{RM.Decode}$ outputs the closest codeword,
each $\vec{y}^{(i)}$ is uniformly distributed over its corresponding symbol set.
Interpreting $\vec{y}^{(i)}$ as an element of $\GF{q}$,
the resulting vector $\vec{y}\in\GF{q}^{n_e}$ is uniformly random.

Let $\mathcal{C}$ denote the RS$[n_e,k_e]$ code.
The code contains exactly $|\mathcal{C}| = q^{k_e}$ codewords.
For any codeword $\vec{c}\in\mathcal{C}$, define the Hamming ball of radius $t_e$ as
\[
B(\vec{c},t_e)
= \{\vec{v}\in\GF{q}^{n_e} : d_H(\vec{v},\vec{c}) \le t_e\}.
\]
The size of each such ball is
\[
|B(\vec{c},t_e)|
= \sum_{i=0}^{t_e} \binom{n_e}{i}(q-1)^i.
\]

Since RS codes achieve a minimum distance $d_e$,
the Hamming balls $\{B(\vec{c},t_e)\}_{\vec{c}\in\mathcal{C}}$
are pairwise disjoint.
Therefore, the total number of vectors in $\GF{q}^{n_e}$
that can be successfully decoded equals
\[
|\mathcal{C}|\cdot |B(\vec{0},t_e)|
= q^{k_e}\sum_{i=0}^{t_e} \binom{n_e}{i}(q-1)^i.
\]

As the ambient space contains $q^{n_e}$ vectors in total,
the probability that a uniformly random vector $\vec{x}$
falls into one of these decoding regions is
\[
\mathrm{FPR}
= \frac{q^{k_e}}{q^{n_e}}
\sum_{i=0}^{t_e} \binom{n_e}{i}(q-1)^i,    
\]
which completes the proof.
\end{proof}
\end{theorem}

By substituting the FPRs of $\mathcal{C}_1$ and $\mathcal{C}_2$ into Eq.~\eqref{equ:fpr}, the overall FPR of \cecinst is expressed as their product.
We will provide the specific FPR of \cecinst under concrete parameters in~\Cref{sec:sound_cecmark}, where we also provide a comparative study with other existing schemes. 
As demonstrated in the evaluation, the concatenated structure of \cecinst effectively maintains a low FPR while achieving high error-correction robustness.

\section{Application: CEC-based Watermarks}
\label{sect:experiments}
Due to the favorable properties of CECs, we present their use in watermarking for AIGCs as a representative application. 

\subsection{CECMark}
\label{subsec:cecmark}

We propose CECMark, an in-processing watermarking scheme for diffusion models. 
Specifically, CECMark embeds watermark information in the initial latents using \cecinst, where an auxiliary mask $\vec{mask}_1$ is applied to improve the quality of the generated outputs.
The formal procedures are presented in Algorithm~\ref{alg:watermark_keygen}, Algorithm~\ref{alg:watermark_embedding} and Algorithm~\ref{alg:watermark_extraction}. 
Illustrations of CECMark are provided in Appendix~\ref{sect:cecmark_embed_extract}.

\mypara{Watermark Key Generation}
Let $N$ be the length of the initial latent and $[N]=\{i \mid 1 \leq i \leq N\}$.
In Algorithm~\ref{alg:watermark_keygen}, the master key of \cecinst is generated in Line~1, while a uniformly random permutation $\mat{P}$ over $[N]$ is sampled in Line~2. 
The master key must remain secret to ensure security, whereas $\mat{P}$ is random but does not need to be kept confidential.

\mypara{Watermark Embedding}
In Algorithm~\ref{alg:watermark_embedding}, $\mathsf{CECMark}.\mathsf{Embed}$ first encrypts $\vec{m}$ to $\vec{c}$ via \cecinst (Line~1). 
Given an initial latent vector $\vec{z}$, we sort its coordinates by absolute value and apply a permutation that maps the $n$ largest-magnitude positions to the prefix (Line~2-4). 
This step follows the optimization strategy proposed in~\cite{yang2025t2smark}, which restricts embedding to coordinates that remain stable under noise perturbations. 
We then embed $\mathbf{c}$ by modulating the signs of these $n$ coordinates (Line~5).
Finally, a global permutation $\mat{P}$ is applied to decorrelate the embedding locations (Line~6). 
Without this step, localized perturbations in the image domain may induce burst errors in the latent space, violating the independence assumption commonly used in DFR analysis. 
The global permutation mitigates such structured errors and improves robustness.

\mypara{Watermark Extraction}
As shown in Algorithm~\ref{alg:watermark_extraction}, given a recovered latent vector, $\mathsf{CECMark}.\mathsf{Extract}(\cdot)$ first applies the inverse permutation $\mat{P}^{-1}$ to restore the canonical ordering (Line~1). 
The watermark bits are then directly recovered from the signs of the first $n$ coordinates in Line~2. 
In Line~3, the resulting vector is decrypted by \cecinst, yielding either the embedded message or the failure symbol $\perp$. 

\begin{algorithm}[t]
\caption{$\mathsf{CECMark}.\mathsf{KeyGen}$}
\label{alg:watermark_keygen}
\begin{algorithmic}[1]
\Require{Security parameter $\lambda$}
\Ensure{Master key $\vec{mk} \in \GF{2}^{2\lambda_1}$, permutation $\mat{P} : [N] \rightarrow [N]$}

\State $\vec{mk} \leftarrow \cecinst.\mathsf{KeyGen}(1^\lambda)$
\State $\mat{P} \xleftarrow{\$} \mathcal{P}_N$ \COMMENT{$\mathcal{P}_N$ is the set of all bijections over $[N]$}
\State \Return $\vec{mk}, \mat{P}$

\end{algorithmic}
\end{algorithm}

\begin{algorithm}[t]
\caption{$\mathsf{CECMark}.\mathsf{Embed}$}
\label{alg:watermark_embedding}
\begin{algorithmic}[1]

\Require{Security parameter $\lambda$, initial latent $\mathbf{z} \in \mathbb{R}^N$, message $\vec{m} \in \GF{2}^k$, master key $\vec{mk} \in \GF{2}^{2\lambda_1}$, system permutation $\mathbf{P} \in \mathcal{P}_N$}
\Ensure{Watermarked latent $\mathbf{z}' \in \mathbb{R}^N$}

\State $\vec{c} \leftarrow \cecinst.\mathsf{Encrypt}(1^\lambda,\vec{mk},\vec{m})$ \COMMENT{Ciphertext of length $n$}

\State $\vec{s} \xleftarrow{\$} c\parallel \mathsf{RandBits}(N-n)$ \COMMENT{Pad to length $N$}

\State $\mathbf{v} \gets \mathsf{Sort}(|\mathbf{z}|, \mathsf{descending})$ \COMMENT{Sort magnitude}
\State $\mathbf{v} \xleftarrow{\$} \mathsf{Shuffle}(\mathbf{v}_{1:n})\parallel \mathsf{Shuffle}(\mathbf{v}_{n+1:N})$

\COMMENT{Shuffle significant and insignificant parts}

\State $\mathbf{z}' \gets  (2\mathbf{s} - 1)\cdot \mathbf{v}$ \COMMENT{Embed into signs}

\State $\mathbf{z}' \gets \mathbf{P}(\mathbf{z}')$ \COMMENT{Apply the global permutation}

\State \Return $\mathbf{z}'$

\end{algorithmic}
\end{algorithm}

\begin{algorithm}[t]
\caption{$\mathsf{CECMark}.\mathsf{Extract}$}
\label{alg:watermark_extraction}
\begin{algorithmic}[1]

\Require{Security parameter $\lambda$, recovered latent $\hat{\mathbf{z}} \in \mathbb{R}^N$, master key $\vec{mk} \in \GF{2}^{2\lambda_1}$, system permutation $\mathbf{P} \in \mathcal{P}_N$}
\Ensure{Extracted message vector $\hat{\vec{m}} \in \GF{2}^k$ \textbf{or} $\mathsf{\perp}$}

\State $\hat{\mathbf{z}} \gets \mathbf{P}^{-1} (\hat{\mathbf{z}})$ \COMMENT{Reverse the global permutation}
\State $\hat{\mathbf{c}} \gets \mathsf{Sign}(\hat{\mathbf{z}}_{1:n})/2+0.5$ \COMMENT{Extract ciphertext}

\State $\hat{\vec{m}} \leftarrow \cecinst.\mathsf{Decrypt}(1^\lambda,\vec{mk},\hat{\mathbf{c}})$ \COMMENT{Message or $\perp$}
\State \Return $\hat{\vec{m}}$

\end{algorithmic}
\end{algorithm}

\subsection{Experimental Setups}
\label{sect:experiment_setups}
\mypara{Text-to-image Model}
We use Stable-Diffusion-2.1-Base~\cite{Rombach_2022_CVPR} to implement watermarking schemes. The size of the initial latent is $n = 4 \times 64 \times 64$ and the output image size is $512 \times 512$.
For image generation, we set generation steps to $T = 50$ with DPMSolver~\cite{Lu2025DPMSolver} under a guidance scale of $\gamma = 3.0$.
For latent inversion, we adopt the DPMInversionSolver~\cite{DPMInv} for $T' = 50$ steps under a null prompt condition.
We sample $1,000$ prompts from the test split of the Stable Diffusion Prompts Dataset~\cite{StableDiffusionPromptsDatasets} to generate images. 

\mypara{Parameters of CECMark}
We set the security parameters as $k = \lambda_1 = 256$ and $\lambda_2 = 64$. This configuration aligns with NIST Category 5 to ensure the highest post-quantum security level while providing sufficient capacity for practical payloads, with the $2^{64}$ randomness space being more than sufficient throughout the watermark's lifecycle.
This is consistent with standard practice in symmetric encryption constructions built from block ciphers such as AES~\cite{daemen2002design}, where initialization vectors or nonces are typically much shorter than the underlying security parameter.
We instantiate the pseudorandom functions $f_1$ and $f_2$ using \textsf{ChaCha20}'s keystream generator~\cite{bernstein2008chacha}, which is parameterized by a 256-bit key and a 96-bit nonce, i.e., we set 
\begin{align*}
    f_1&=\mathsf{ChaCha20.KeyStreamGenerator}(\vec{mk}_1,1^{96}),\\
    f_2&=\mathsf{ChaCha20.KeyStreamGenerator}(\vec{mk}_2,0^{32}\parallel \theta).
\end{align*}
To evaluate the impact of embedding capacity and robustness, we further instantiate three concrete watermark configurations with different channel lengths, as summarized in~\Cref{tab:rssmark-params}.

\begin{table}[t]
\centering
\caption{Parameters of CECMark Variants.}
\label{tab:rssmark-params}
\small
\begin{tabular}{cccc}
\toprule
 & CECMark-0 & CECMark-1 & CECMark-2 \\
\midrule
$n_1$      & $4,096$ & $2,560$  & $2,048$ \\
$n_2$      & $12,288$ & $9,728$  & $6,144$  \\
RS$_1$     & $[16,8,9]$ & $[10,8,3]$   & $[16,8,9]$ \\
RM$_1$     & $[256,8,128]$ & $[256,8,128]$ & $[128,8,64]$ \\
RS$_2$     & $[48,32,17]$ & $[38,32,7]$  & $[48,32,17]$ \\
RM$_2$     & $[256,8,128]$ & $[256,8,128]$ & $[128,8,64]$ \\
Occupancy  & $100\%$ & $75\%$ & $50\%$ \\
\bottomrule
\end{tabular}
\end{table}

\mypara{Efficiency of CECMark}
CECMark achieves 243--250 it/s (iterations per second) throughput for encoding and 7.79--10.46 it/s for decoding, as RM-RS decoding is more complex than encoding.
This asymmetry is acceptable since encoding is much more frequent in practice.
Moreover, compared to image generation (0.56 it/s) and latent inversion (0.85 it/s), the overhead of CECMark is negligible.
The detailed analysis is in Appendix~\ref{sec:throughput}.

\mypara{Baseline Watermarking Schemes}
We adopt three high-capacity in-processing watermarking schemes as baselines: GS~\cite{yang2024gaussian}, PRCMark~\cite{DBLP:conf/iclr/GunnZS25}, and T2SMark~\cite{yang2025t2smark}. 
Since the performance of watermark schemes is closely tied to their encoding capacity, we fix the payload size to 256 bits across all schemes to ensure a fair comparison.
All payloads are mutually distinct across generations.
Detailed settings of each scheme are in Appendix~\ref{sect:baseline-details}.

\subsection{Fidelity of CECMark}
\label{sec:filde_cecmark}

\mypara{Image Quality Metrics}
We quantitatively evaluate image quality using Fréchet Inception Distance (FID)~\cite{NIPS2017_8a1d6947FID} and CLIP score~\cite{hessel-etal-2021-clipscore}.
FID measures the distributional discrepancy between reference and evaluated images, with lower values indicating closer similarity.
We fix the reference set to 1,000 images generated by the original pipeline.
CLIP score measures the semantic alignment between the generated images and their prompts, where higher values reflect stronger fidelity. 
We employ a 1.3B CLIP model, CLIP-ViT-g-14-laion2B-s12B-b42K~\cite{Schuhmann2022CLIP}, to evaluate similarity for each image-prompt pair and compute the average across all pairs.

\mypara{Consistency Metrics}
To assess the consistency between unwatermarked and watermarked images, we compute FID and CLIP scores for the original pipeline and all watermarking schemes over 10 trials, each with a fresh seed or keypair. 
We then perform two-sample t-tests with hypotheses $H_0: \mu_1 = \mu_0$ and $H_1: \mu_1 \ne \mu_0$, where $\mu_0$ and $\mu_1$ denote the mean metric values for unwatermarked and watermarked images. 
A higher $p$-value provides stronger support for $H_0$.

\mypara{Evaluation Results}
The evaluation results are shown in~\Cref{tab:quality}, and the images generated by different schemes are shown in~\Cref{fig:image-examples}.
For FID, all watermark schemes yield $p$-values below 0.05, indicating that watermarking alters the distributional similarity compared to the reference watermarked images.
However, CECMark achieves the highest $p$-value among all schemes, suggesting greater consistency with unwatermarked images.
For the CLIP score, all schemes show no statistically significant difference from the baseline except GS.
Notably, CECMark-0 demonstrates both the highest semantic fidelity and the strongest consistency with unwatermarked images. 
Overall, CECMark preserves semantic alignment while exerting a smaller impact on distributional similarity compared to other watermarking schemes.

\begin{table}[t]
\caption{Quality of watermarked images. We report the mean and standard error (SE) of each metric, together with the $p$-value between watermarked and non-watermarked images.}
\label{tab:quality}
\setlength{\tabcolsep}{2pt}
\small
\begin{tabular}{cccc}
\toprule
Schemes & FID $\downarrow$ ($p$-value $\uparrow$)  & CLIP Score $\uparrow$ ($p$-value $\uparrow$) \\ \midrule
w./o. WM & $50.413\pm0.088 (\phantom{1}\text{ --- }\phantom{1})$ & $34.637\pm0.026 (\phantom{1}\text{ --- }\phantom{1})$ \\ \midrule
GS & $53.471\pm0.244 (0.000)$ & \textbf{$34.521\pm0.048 (0.047)$} \\
PRCMark & $\textbf{50.853}\pm0.093 (0.003)$ & $\underline{34.593}\pm0.038 (\underline{0.359})$ \\
T2SMark & $51.269\pm0.096 (0.000)$ & $34.586\pm0.042 (0.321)$ \\ \midrule
CECMark-0 & $50.965\pm0.107 (0.001)$ & $\textbf{34.617}\pm0.023 (\textbf{0.576})$ \\
CECMark-1 & $\underline{50.866}\pm0.145 (\textbf{0.015})$ & $34.541\pm0.045 (0.085)$ \\
CECMark-2 & \textbf{$50.986\pm0.146 (\underline{0.003})$} & $34.578\pm0.020 (0.091)$ \\  \bottomrule
\end{tabular}
\end{table}

\subsection{Confidentiality of CECMark}
\label{sec:conf_cecmark}

\mypara{Theoretical Analysis}
The confidentiality of different watermarking schemes is summarized in Table~\ref{tab:complexity-comparison}.
(1) In GS, there exists no randomness\footnote{
Although~\cite{thietke2025towards} discusses some GS variants that could achieve IND-CPA security, they require recording all previously used nonces and performing decryption by enumerating them, which is highly impractical in realistic settings.},
i.e., when the message is fixed, the signs of the initial latent are fixed. 
In the IND-CPA game, we can simply query the oracle to encrypt $\vec{m}_0$ again and obtain the ciphertext $\vec{c}_0$.
If $\vec{c}^{*}=\vec{c}_0$, output $b=0$.
Otherwise, output $b=1$.
The success probability is $100\%$ with $O(1)$ queries, i.e., $O(n)$ complexity.
(2) As discussed in~\cite{wang2025cryptanalysis}, PRCMark suffers from a $\binom{n}{t/2}$ key recovery attack.
By modifying the key recovery attack, the adversary can obtain a partial secret key and further decode the message with accuracy comparable to that of the secret-key holder. 
The details of the attack are provided in Appendix~\ref{sect:prc_attack}.
(3) Since the length of the master key $\vec{mk}_{\rm T2S}$ of T2SMark is only $16$ bits, an adversary can directly enumerate the master key and recover the session key to break the scheme.
The adversary could also directly enumerate the $16$-bit session key $\vec{sk}_{\rm T2S}$ to recover the message.
Therefore, the security of T2SMark is determined by the minimum length of the master key and the session key. 
(4) The security of CECMark depends only on the length of the master key, while the nonce segment of the latent representation mainly serves to provide diversity and prevent reuse.

\mypara{Empirical Distinguishability}
To evaluate real-world confidentiality, following~\cite{DBLP:conf/iclr/GunnZS25,yang2025t2smark}, we train a ResNet-18 binary classifier for each scheme to distinguish watermarked images from unwatermarked images.
The results show that CECMark-0 achieves the lowest distinguishability among all watermark schemes.
The detailed settings and results are provided in Appendix~\ref{sect:hyper}.

\begin{table}[t]
\centering
\caption{Comparison of theoretical and concrete time complexities against confidentiality across different schemes.}
\label{tab:complexity-comparison}
\setlength{\tabcolsep}{10pt}
\small
\begin{tabular}{lccc}
\toprule
\textbf{Scheme} & \textbf{Theoretical Time} & \textbf{Concrete Time} \\
\midrule
GS        & $O(n)$ & $2^{14}$  \\
PRCMark       & $O(\binom{n}{t/2})$ & $2^{28}$ \\
T2SMark   & $O(2^{\mathsf{min}(|\vec{mk}_{\rm T2S}|,|\vec{sk}_{\rm T2S}|)})$ & $2^{16}$ \\
\midrule
CECMark   & $O(2^{\lambda_1})$  & $2^{256}$ \\
\bottomrule
\end{tabular}
\end{table}

\begin{table}[t]
\caption{Robustness of watermarking schemes measured by True Positive Rate (TPR). 
}
\label{tab:tpr}
\setlength{\tabcolsep}{10pt}
\small
\begin{tabular}{ccc}
\toprule
Schemes & TPR@Clean $\uparrow$ & TPR@Perturbed $\uparrow$\\ \midrule

GS & $0.9689\pm0.0431$ & 0.4959 \\
PRCMark & $0.9994\pm0.0007$ & 0.7277 \\
T2SMark & $0.9863\pm0.0047$ & 0.5197 \\ \midrule
CECMark-0 & $\underline{0.9998}\pm0.0004$ & 0.8027 \\
CECMark-1 & $\textbf{1.0000}\pm0.0000$ & \textbf{0.8339} \\
CECMark-2 & $\textbf{1.0000}\pm0.0000$ & \underline{0.8148} \\ \bottomrule
\end{tabular}
\end{table}

\subsection{Robustness of CECMark}
\label{sec:robust_cecmark}

\mypara{Evaluation Metric}
We evaluate the robustness of watermarking schemes using the True Positive Rate (TPR), 
defined as the proportion of watermarked images that are still correctly identified 
as watermarked after undergoing various perturbations.
Note that the verification is evaluated in the multi-bit setting, \textit{where a true positive is counted only if the decoded message perfectly equals the original embedded message}.

\mypara{Evaluation Results}  
As a baseline, we first evaluate the clean TPR on 1,000 unperturbed watermarked images.
This evaluation is repeated across 10 trials, each with a fresh keypair.
We then evaluate robustness under 56 perturbations spanning diverse transformations and strength levels.
Detailed settings are explained in Appendix~\ref{sect:pert}.
The overall clean and perturbed TPRs are reported in~\Cref{tab:tpr} and the per-perturbation results are provided in~\Cref{fig:pert-image-tpr}.
Evaluation results demonstrate that CECMark consistently achieves higher robustness than existing watermark schemes.
Among the three parameter sets, CECMark-1 achieves the highest robustness, while CECMark-0 performs the lowest, highlighting the effectiveness of the selective embedding strategy.

\mypara{Decoding Boundary Analysis}
To characterize the decoding boundaries of different schemes, we employ the ratio of flipped signs between original and inverted latents as a unified measure of perturbation strength.
We report the corresponding TPRs under varying sign-flip ratios in~\Cref{fig:overall-sign-flip}.
The results show that CECMark maintains a stable TPR until the sign-flip ratio reaches 37\%, and its TPR drops to zero around 40\%.
In contrast, the strongest baseline PRCMark begins to degrade at 33\% and drops to zero near 39\%.
Consequently, CECMark achieves a decoding boundary approximately 10\% larger than the existing schemes, thereby achieving better watermark robustness.

\begin{figure}[t]
\centering
\includegraphics[width=0.7\linewidth]{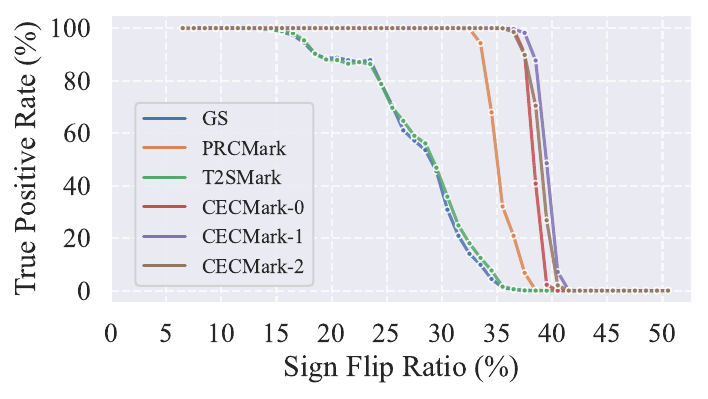}
\caption{True Positive Rate (TPR) of watermarks under different latent sign flip ratios.}
\label{fig:overall-sign-flip}
\end{figure}

\subsection{Soundness of CECMark}
\label{sec:sound_cecmark}

For the baseline watermarking schemes, the soundness is either not achieved or costly.
We summarize the soundness of different schemes in~\Cref{tab:soundness-comparison}.
Empirically, we observe zero false positives across 10,000 unwatermarked images for all watermarking schemes.
However, CECMark achieves a low theoretical FPR without introducing any other side effects.

\mypara{Soundness of GS}
GS works in either detection mode or decoding mode.
In detection mode, GS encodes a fixed message and detects whether the watermark exists with a threshold.
In decoding mode, even if the input image were not watermarked, the decoding result would be a valid message.
Therefore, soundness is achieved only in the detection version of GS.

\mypara{Soundness of PRCMark}
PRCMark chooses extra parity bits to bound the FPR, i.e., after decoding, the result is accepted if and only if the parity check bits are identical to the preset value. 
While this approach guarantees a fixed FPR, it inherently increases the code rate. 
In Belief Propagation (BP)-based decoders, a higher code rate typically reduces the error-correction capability, thereby undermining the overall robustness of the scheme.


\mypara{Soundness of T2SMark}
Because the decoding method for duplicated codes could be applied to any codeword, T2SMark has to calculate a confidence level to detect the watermark.
A high threshold weakens robustness by making the watermark easier to remove, whereas a low threshold increases false positives on unwatermarked images.
T2SMark therefore faces an inherent trade-off between robustness and soundness, which can be regarded as a side effect.

\begin{table}[b]
\centering
\caption{Comparison of soundness, side effect, and target FPR across different watermarking schemes.}
\label{tab:soundness-comparison}
\setlength{\tabcolsep}{1pt}
\small
\begin{tabular}{lccc}
\toprule
\textbf{Scheme} & \textbf{Soundness} & 
\textbf{Side Effect} &
\textbf{Target FPR} \\
\midrule
GS        & \rc (Detection only) & - & -  \\
PRCMark   & \Cfill  & Reduced robustness & $10^{-5}$ \\
T2SMark   & \Cfill  & Reduced robustness & $10^{-6}$ \\ \midrule
CECMark-0 & \Cfill  & None & $10^{-17}$ \\
CECMark-1 & \Cfill  & None & $10^{-5}$  \\
CECMark-2 & \Cfill  & None & $10^{-17}$ \\
\bottomrule
\end{tabular}
\end{table}


\section{Discussions}
\label{sect:discussion}
\subsection{Comparison with Duplicated Codes}
\label{subsec:dfr_dup_rmrs}

As duplicated codes are applied in GS and T2SMark, we carry out a theoretical analysis of the DFR for duplicated codes.

\begin{theorem}[DFR of Duplicated Codes]
    Let $n$ denote the code length, $k$ the message length, and $d=n/k$ the duplication factor.
    Let $p$ be the noise rate of the noisy channel. 
    The DFR of the duplicated code is as follows.
    
    \noindent If $d$ is odd, the DFR is 
    \[
    1-\bigg[1-\sum_{i=\lfloor d/2 \rfloor +1}^{d}
    \binom{d}{i}p^{i}(1-p)^{d-i} \bigg]^{k}.
    \]
    If $d$ is even, the DFR is
    \[
    1-\bigg[1-\sum_{i=d/2+1}^{d}
    \binom{d}{i}p^{i}(1-p)^{d-i}
    -\frac{1}{2}\binom{d}{d/2}p^{d/2}(1-p)^{d/2}
    \bigg]^{k}.
    \]
    \begin{proof}
        Each message bit is repeated $d=n/k$ times.
        When $d$ is odd, decoding fails if and only if more than half of its copies are flipped.
    The corresponding probability is
    \[
    p_0=\sum_{i=\lfloor d/2 \rfloor +1}^{d}\binom{d}{i}p^{i}(1-p)^{d-i}.
    \]
    When $d$ is even, in addition to the above case, a decoding failure also occurs with $1/2$ probability if exactly half of the positions are flipped.
    The corresponding probability is
    \[
    p_0=\sum_{i=\lfloor d/2 \rfloor +1}^{d}\binom{d}{i}p^{i}(1-p)^{d-i}
    +\frac{1}{2}\binom{d}{d/2}p^{d/2}(1-p)^{d/2}.
    \]

    A decoding failure occurs if and only if at least one of the duplicated codes is decoded incorrectly.
    Thus,
    \[
    \mathsf{DFR}=1-(1-p_0)^k,
    \]
    which completes the proof.
    \end{proof}
\end{theorem}

\begin{figure}[t]
\centering
\includegraphics[width=\linewidth]{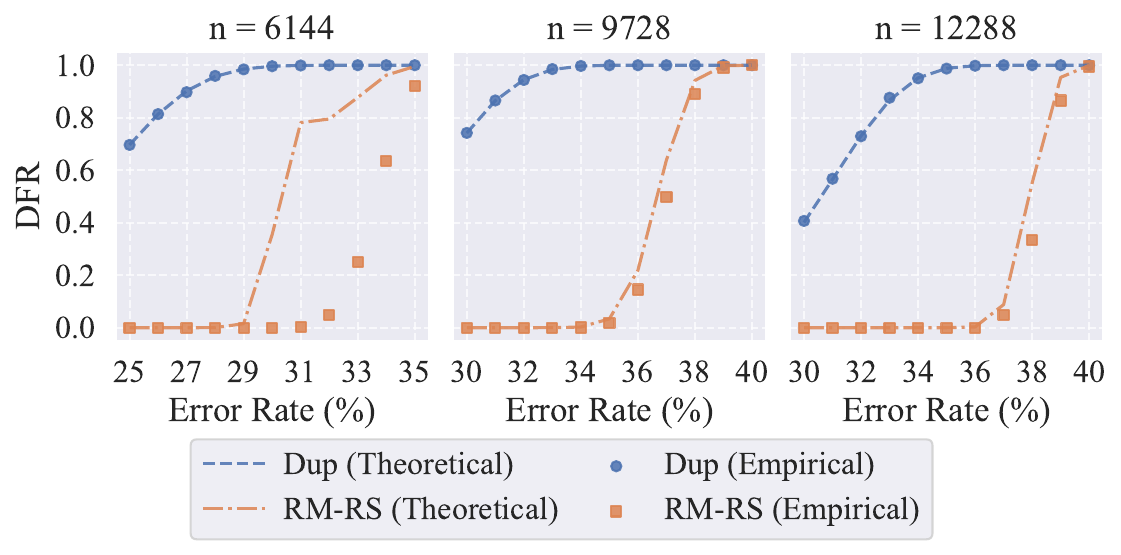}
\caption{DFR comparison between duplicated codes and RM-RS codes with different $n$.}
\label{fig:dfr_all}
\end{figure}



  

\mypara{DFR Comparison Results}
We then compare the DFR of the RM-RS code used in our CECMark scheme with that of the duplicated code adopted in GS and T2SMark. 
For a fixed message length and code length, we vary the channel error rate and evaluate the resulting DFR over $10^4$ trials for each setting. 
Specifically, we consider the message length $k=256$ and the code length $n \in \{6144, 9728, 12288\}$, corresponding to the parameters $k,n_2$ in CECMark-2, CECMark-1, and CECMark-0, respectively.
As shown in~\Cref{fig:dfr_all}, for the duplicated code, the theoretical DFR matches the empirical results exactly.
For the RM-RS code, the theoretical DFR serves as an upper bound across all parameters, indicating that the analysis is conservative and the actual performance is better than predicted.
Across all tested error rates, the RM-RS code consistently achieves a lower DFR than the duplicated code, both theoretically and empirically.

\subsection{Comparison with PRCs}
\label{sec:comp_prc}

Pseudorandom error-correcting Codes (PRCs)~\cite{christ2024pseudorandom} is the most closely related cryptographic primitive to CECs and is also a promising primitive for AIGC watermarking.
We therefore provide a dedicated comparison with PRCs to highlight the advantages of CECs and explain why CECs are needed as a distinct primitive.

\mypara{Property-Level Comparison}
PRCs~\cite{christ2024pseudorandom} are required to simultaneously satisfy three properties: undetectability, robustness, and soundness. 
These properties closely align with the requirements of CECs, except that PRCs rely on undetectability whereas CECs rely on confidentiality.
In the context of watermarking, undetectability serves two roles: it helps preserve the perceptual quality of generated content and also provides a form of information hiding.
However, undetectability is a strong requirement, and existing PRC instantiations have been shown to lose this property under practical parameters~\cite{wang2025cryptanalysis}.
Moreover, as indicated in~\Cref{sec:robust_cecmark}, enforcing undetectability may impose additional constraints on robustness.
In contrast, CECs formalize information protection directly through confidentiality.
Meanwhile, \Cref{sec:filde_cecmark} shows that CECMark can also preserve high perceptual quality in practice.

\mypara{Instantiability-Level Comparison}
As far as we know, although the works in~\cite{christ2024pseudorandom,DBLP:conf/iclr/GunnZS25,ghentiyala2024new,christ2025improved,moitra2024edit} claim that they can successfully construct PRCs, all of these constructions have flaws in choosing concrete parameters or fail to balance robustness and undetectability. 
In Appendix~\ref{sect:inst_prc}, we will conduct a detailed analysis of their limitations.
Overall, while PRCs are theoretically attractive, their practical performance is still limited.

\mypara{Building PRCs based on CECs}
In \cecinst, the RM-RS code is inherently distinguishable from random due to its strong algebraic structure, and thus does not provide undetectability. 
In our framework, undetectability depends on whether the underlying code is computationally indistinguishable from random. 
While one may replace RM-RS with pseudorandom codes such as LDPC or MDPC combined with noise, prior attacks~\cite{wang2025cryptanalysis} indicate that achieving both strong undetectability and robustness in practice remains challenging. 
Therefore, achieving full PRC security with suitable codes remains an interesting direction for future work.

\subsection{Other Possible Applications}
Following~\cite{christ2024pseudorandom,DBLP:conf/iclr/GunnZS25}, we regard watermarking as our main application scenario.
The primary challenge in applying CECs to Large Language Model (LLM) watermarking is how to embed the \cecinst codewords into token sequences.
A common approach adopted by PRCs~\cite{christ2024pseudorandom} is to let each watermark bit influence the sampling procedure of individual tokens.
However, as analyzed in~\cite{wang2025cryptanalysis}, under reasonable text quality constraints, the entropy of output tokens is too low, resulting in a noise rate that remains an unsolved challenge for existing error-correction decoders.
Moreover, constructing an edit-robust watermark for LLMs often incurs prohibitive overhead. 
Specifically, existing schemes like~\cite{christ2024pseudorandom} employ majority encoding for each watermark bit to resist the deletion channel, which substantially expands the watermark length and renders the construction impractical for deployment in real-world LLMs. 
While \cecinst could theoretically be extended to resist deletion and insertion errors via similar majority encoding techniques, we currently restrict its application to GIMs due to these practical constraints. 
We leave the concrete optimization and implementation of CECs for LLM watermarking as future work.
We believe that CECs may also be applicable to more applications such as watermarking video diffusion models~\cite{hu2025videomark}, which represents a promising direction for future research.

\section{Conclusion}
\label{sect:conclusion}
We present a new primitive, Confidential Error-Correcting Codes (CECs), a symmetric-key primitive that integrates encryption with error-correcting capability.
The confidentiality, robustness, and soundness of CECs are designed to address the key challenges in watermarking generative models.
We further develop a CEC construction framework based on a randomness-security separation structure, reducing confidentiality to pseudorandom functions, while ensuring robustness and soundness with error-correcting codes.
Instantiated with concatenated codes, we present \cecinst and analyze its security properties theoretically.
Finally, we design CECMark, a watermarking scheme for generative image models based on \cecinst.
Experiments show that CECMark maintains strong robustness and soundness and achieves consistently higher detection accuracy than existing methods under various image perturbations.
Looking forward, we believe that CECs can serve as a foundational building block for trustworthy watermarking systems, enabling the field to move from heuristic designs toward cryptographically grounded and practically reliable solutions.

\clearpage

\bibliographystyle{plainurl}
\bibliography{ref_prc}

@InProceedings{Tancik_2020_CVPR,
author = {Tancik, Matthew and Mildenhall, Ben and Ng, Ren},
title = {StegaStamp: Invisible Hyperlinks in Physical Photographs},
booktitle = {Proceedings of the IEEE/CVF Conference on Computer Vision and Pattern Recognition (CVPR)},
month = {June},
year = {2020}
}

@InProceedings{pmlr-v235-an24a,
  title = 	 {{WAVES}: Benchmarking the Robustness of Image Watermarks},
  author =       {An, Bang and Ding, Mucong and Rabbani, Tahseen and Agrawal, Aakriti and Xu, Yuancheng and Deng, Chenghao and Zhu, Sicheng and Mohamed, Abdirisak and Wen, Yuxin and Goldstein, Tom and Huang, Furong},
  booktitle = 	 {Proceedings of the 41st International Conference on Machine Learning},
  pages = 	 {1456--1492},
  year = 	 {2024},
  editor = 	 {Salakhutdinov, Ruslan and Kolter, Zico and Heller, Katherine and Weller, Adrian and Oliver, Nuria and Scarlett, Jonathan and Berkenkamp, Felix},
  volume = 	 {235},
  series = 	 {Proceedings of Machine Learning Research},
  month = 	 {21--27 Jul},
  publisher =    {PMLR},
  url = 	 {https://proceedings.mlr.press/v235/an24a.html}
}

@article{Lu2025DPMSolver,
author="Lu, Cheng
and Zhou, Yuhao
and Bao, Fan
and Chen, Jianfei
and Li, Chongxuan
and Zhu, Jun",
title="DPM-Solver++: Fast Solver for Guided Sampling of Diffusion Probabilistic Models",
journal="Machine Intelligence Research",
year="2025",
month="Aug",
day="01",
volume="22",
number="4",
pages="730--751",
issn="2731-5398",


}

@inproceedings{
yang2025t2smark,
title={T2{SM}ark: Balancing Robustness and Diversity in Noise-as-Watermark for Diffusion Models},
author={Jindong Yang and Han Fang and Weiming Zhang and Nenghai Yu and Kejiang Chen},
booktitle={The Thirty-ninth Annual Conference on Neural Information Processing Systems},
year={2025},

}

@misc{gowal2025synthidimageimagewatermarkinginternet,
      title={SynthID-Image: Image watermarking at internet scale}, 
      author={Sven Gowal and Rudy Bunel and Florian Stimberg and David Stutz and Guillermo Ortiz-Jimenez and Christina Kouridi and Mel Vecerik and Jamie Hayes and Sylvestre-Alvise Rebuffi and Paul Bernard and Chris Gamble and Miklós Z. Horváth and Fabian Kaczmarczyck and Alex Kaskasoli and Aleksandar Petrov and Ilia Shumailov and Meghana Thotakuri and Olivia Wiles and Jessica Yung and Zahra Ahmed and Victor Martin and Simon Rosen and Christopher Savčak and Armin Senoner and Nidhi Vyas and Pushmeet Kohli},
      year={2025},
      eprint={2510.09263},
      archivePrefix={arXiv},
      primaryClass={cs.CR},
 
}

@misc{cryptoeprint:2025/1620prclimits,
      author = {Danilo Francati and Yevin Nikhel Goonatilake and Shubham Pawar and Daniele Venturi and Giuseppe Ateniese},
      title = {The Coding Limits of Robust Watermarking for Generative Models},
      howpublished = {Cryptology {ePrint} Archive, Paper 2025/1620},
      year = {2025},
}

@inproceedings{Schuhmann2022CLIP,
title={{LAION}-5{B}: An open large-scale dataset for training next generation image-text models},
author={Christoph Schuhmann and Romain Beaumont and Richard Vencu and Cade W Gordon and Ross Wightman and Mehdi Cherti and Theo Coombes and Aarush Katta and Clayton Mullis and Mitchell Wortsman and Patrick Schramowski and Srivatsa R Kundurthy and Katherine Crowson and Ludwig Schmidt and Robert Kaczmarczyk and Jenia Jitsev},
booktitle={Thirty-sixth Conference on Advances in Neural Information Processing Systems (NeurIPS), Datasets and Benchmarks Track},
year={2022},
volume={35},
pages={25278--25294},

}

@misc{ShieldMntInvisibleWatermark,
  author       = {ShieldMnt},
  title        = {invisible-watermark: python library for invisible image watermark},
  year         = {2021},
  publisher    = {GitHub},
  
}

@inproceedings{hessel-etal-2021-clipscore,
    title = "{CLIPS}core: A Reference-free Evaluation Metric for Image Captioning",
    author = "Hessel, Jack  and
      Holtzman, Ari  and
      Forbes, Maxwell  and
      Le Bras, Ronan  and
      Choi, Yejin",
    editor = "Moens, Marie-Francine  and
      Huang, Xuanjing  and
      Specia, Lucia  and
      Yih, Scott Wen-tau",
    booktitle = "Proceedings of the 2021 Conference on Empirical Methods in Natural Language Processing",
    month = nov,
    year = "2021",
    address = "Online and Punta Cana, Dominican Republic",
    publisher = "Association for Computational Linguistics",

    doi = "10.18653/v1/2021.emnlp-main.595",
    pages = "7514--7528"
}

@inproceedings{NIPS2017_8a1d6947FID,
 author = {Heusel, Martin and Ramsauer, Hubert and Unterthiner, Thomas and Nessler, Bernhard and Hochreiter, Sepp},
 booktitle = {Advances in Neural Information Processing Systems},
 editor = {I. Guyon and U. Von Luxburg and S. Bengio and H. Wallach and R. Fergus and S. Vishwanathan and R. Garnett},
 pages = {},
 publisher = {Curran Associates, Inc.},
 title = {GANs Trained by a Two Time-Scale Update Rule Converge to a Local Nash Equilibrium},
 
 volume = {30},
 year = {2017}
}

@INPROCEEDINGS{DPMInv,
  author={Hong, Seongmin and Lee, Kyeonghyun and Jeon, Suh Yoon and Bae, Hyewon and Chun, Se Young},
  booktitle={2024 IEEE/CVF Conference on Computer Vision and Pattern Recognition (CVPR)}, 
  title={On Exact Inversion of DPM-Solvers}, 
  year={2024},
  volume={},
  number={},
  pages={7069-7078},
}

@misc{StableDiffusionPromptsDatasets,
	author = {{G}ustavosta},
	title = {{S}table-{D}iffusion-{P}rompts},
	year = {2022},
    url = {https://huggingface.co/datasets/Gustavosta/Stable-Diffusion-Prompts}
}

@INPROCEEDINGS {sok_zhao_sp25,
author = { Zhao, Xuandong and Gunn, Sam and Christ, Miranda and Fairoze, Jaiden and Fabrega, Andres and Carlini, Nicholas and Garg, Sanjam and Hong, Sanghyun and Nasr, Milad and Tramer, Florian and Jha, Somesh and Li, Lei and Wang, Yu-Xiang and Song, Dawn },
booktitle = { 2025 IEEE Symposium on Security and Privacy (SP) },
title = {{ SoK: Watermarking for AI-Generated Content }},
year = {2025},
volume = {},
ISSN = {2375-1207},
pages = {2621-2639}
}

@inproceedings{christ2024pseudorandom,
  author    = {Christ, Miranda and Gunn, Sam},
  title     = {Pseudorandom Error-Correcting Codes},
  booktitle = {Advances in Cryptology -- CRYPTO 2024},
  editor    = {Reyzin, Leonid and Stebila, Douglas},
  series    = {Lecture Notes in Computer Science},
  volume    = {14925},
  pages     = {325--347},
  publisher = {Springer},
  address   = {Santa Barbara, CA, USA},
  year      = {2024}
}

@inproceedings{DBLP:conf/iclr/GunnZS25,
  author       = {Sam Gunn and
                  Xuandong Zhao and
                  Dawn Song},
  title        = {An Undetectable Watermark for Generative Image Models},
  booktitle    = {The Thirteenth International Conference on Learning Representations,
                  {ICLR} 2025, Singapore, April 24-28, 2025},
  publisher    = {OpenReview.net},
  year         = {2025},
  bibsource    = {dblp computer science bibliography, https://dblp.org}
}

@inproceedings{yang2024gaussian,
  title={Gaussian shading: Provable performance-lossless image watermarking for diffusion models},
  author={Yang, Zijin and Zeng, Kai and Chen, Kejiang and Fang, Han and Zhang, Weiming and Yu, Nenghai},
  booktitle={Proceedings of the IEEE/CVF Conference on Computer Vision and Pattern Recognition},
  pages={12162--12171},
  year={2024}
}

@article{ALBUSAIDI2024100630,
title = {Redefining boundaries in innovation and knowledge domains: Investigating the impact of generative artificial intelligence on copyright and intellectual property rights},
journal = {Journal of Innovation \& Knowledge},
volume = {9},
number = {4},
pages = {100630},
year = {2024},
issn = {2444-569X},
author = {Adil S. Al-Busaidi and Raghu Raman and Laurie Hughes and Mousa Ahmed Albashrawi and Tegwen Malik and Yogesh K. Dwivedi and Thuraiya {Al- Alawi} and Mohammed AlRizeiqi and Gareth Davies and Mark Fenwick and Parul Gupta and Shashikala Gurpur and Apeksha Hooda and Paulius Jurcys and Daryl Lim and Nicola Lucchi and Tanvi Misra and Ramakrishnan Raman and Anuragini Shirish and Paul Walton}
}

@article{Lucchi_2024, title={ChatGPT: A Case Study on Copyright Challenges for Generative Artificial Intelligence Systems}, volume={15}, number={3}, journal={European Journal of Risk Regulation}, author={Lucchi, Nicola}, year={2024}, pages={602–624}}

@InProceedings{Rombach_2022_CVPR,
    author    = {Rombach, Robin and Blattmann, Andreas and Lorenz, Dominik and Esser, Patrick and Ommer, Bj\"orn},
    title     = {High-Resolution Image Synthesis With Latent Diffusion Models},
    booktitle = {Proceedings of the IEEE/CVF Conference on Computer Vision and Pattern Recognition (CVPR)},
    month     = {6},
    year      = {2022},
    pages     = {10684-10695}
}

@misc{euaiact,
title = {{Regulation (EU) 2024/1689 of the European Parliament and of the Council of 13 June 2024 laying down harmonised rules on artificial intelligence and amending Regulations (EC) No 300/2008, (EU) No 167/2013, (EU) No 168/2013, (EU) 2018/858, (EU) 2018/1139 and (EU) 2019/2144 and Directives 2014/90/EU, (EU) 2016/797 and (EU) 2020/1828 (Artificial Intelligence Act)}},
year= {2024},
author = {{European Parliament and Council of the European Union}}
}

@InProceedings{ghentiyala2024new,
  author =	{Ghentiyala, Surendra and Guruswami, Venkatesan},
  title =	{{New Constructions of Pseudorandom Codes}},
  booktitle =	{Approximation, Randomization, and Combinatorial Optimization. Algorithms and Techniques (APPROX/RANDOM 2025)},
  pages =	{54:1--54:22},
  series =	{Leibniz International Proceedings in Informatics (LIPIcs)},
  ISBN =	{978-3-95977-397-3},
  ISSN =	{1868-8969},
  year =	{2025},
  volume =	{353},
  editor =	{Ene, Alina and Chattopadhyay, Eshan},
  publisher =	{Schloss Dagstuhl -- Leibniz-Zentrum f{\"u}r Informatik},
  address =	{Dagstuhl, Germany},
  URL =		{https://drops.dagstuhl.de/entities/document/10.4230/LIPIcs.APPROX/RANDOM.2025.54},
  URN =		{urn:nbn:de:0030-drops-244202},
  doi =		{10.4230/LIPIcs.APPROX/RANDOM.2025.54}
}

@inproceedings{wang2025cryptanalysis,
  title={Cryptanalysis of LDPC-Based Pseudorandom Error-Correcting Codes},
  author={Wang, Tianrui and Wang, Anyu and Cong, Tianshuo and Ran, Delong and Liu, Jinyuan and Wang, Xiaoyun},
  booktitle={35th USENIX Security Symposium (USENIX Security 26)},
  year={2026},
publisher = {USENIX Association},
month = aug
}

@inproceedings{
thietke2025towards,
title={Towards A Correct Usage of Cryptography in Semantic Watermarks for Diffusion Models},
author={Jonas Thietke and Andreas M{\"u}ller and Denis Lukovnikov and Asja Fischer and Erwin Quiring},
booktitle={The 1st Workshop on GenAI Watermarking},
year={2025},
url={https://openreview.net/forum?id=yENK6DYuXv}
}

@article{gaborit2025hamming,
  title={Hamming quasi-cyclic (hqc)},
  author={Gaborit, Philippe and Aguilar-Melchor, Carlos and Aragon, Nicolas and Bettaieb, Slim and Bidoux, Lo{\"\i}c and Blazy, Olivier and Deneuville, Jean-Christophe and Persichetti, Edoardo and Z{\'e}mor, Gilles and Bos, Jurjen and others},
  journal={NIST Post-Quantum Standardization, 4th Round; National Institute of Standards and Technology: Gaithersburg, MD, USA},
  year={2025}
}

@book{barker2007recommendation,
  title={Recommendation for random number generation using deterministic random bit generators (revised)},
  author={Barker, Elaine B and Kelsey, John Michael and others},
  year={2007},
  publisher={US Department of Commerce, Technology Administration, National Institute of Standards and Technology}
}

@book{daemen2002design,
  title={The design of Rijndael},
  author={Daemen, Joan and Rijmen, Vincent},
  volume={2},
  year={2002},
  publisher={Springer}
}

@techreport{krawczyk1997hmac,
  title={HMAC: Keyed-hashing for message authentication},
  author={Krawczyk, Hugo and Bellare, Mihir and Canetti, Ran},
  year={1997}
}

@article{jindal2015rc4,
  title={RC4 encryption-A literature survey},
  author={Jindal, Poonam and Singh, Brahmjit},
  journal={Procedia Computer Science},
  volume={46},
  pages={697--705},
  year={2015},
  publisher={Elsevier}
}

@article{shafiee2016isaac,
  title={ISAAC: A convolutional neural network accelerator with in-situ analog arithmetic in crossbars},
  author={Shafiee, Ali and Nag, Anirban and Muralimanohar, Naveen and Balasubramonian, Rajeev and Strachan, John Paul and Hu, Miao and Williams, R Stanley and Srikumar, Vivek},
  journal={ACM SIGARCH Computer Architecture News},
  volume={44},
  number={3},
  pages={14--26},
  year={2016},
  publisher={ACM New York, NY, USA}
}

@inproceedings{bernstein2008chacha,
  title={ChaCha, a variant of Salsa20},
  author={Bernstein, Daniel J and others},
  booktitle={Workshop record of SASC},
  volume={8},
  number={1},
  pages={3--5},
  year={2008},
  organization={Lausanne, Switzerland}
}

@article{christ2025improved,
  title={Improved Pseudorandom Codes from Permuted Puzzles},
  author={Christ, Miranda and Golowich, Noah and Gunn, Sam and Moitra, Ankur and Wichs, Daniel},
  journal={arXiv preprint arXiv:2512.08918},
  year={2025}
}

@inproceedings{moitra2024edit,
  title={Edit distance robust watermarks for language models},
  author={Moitra, Ankur and Golowich, Noah},
  booktitle={Proceedings of the 38th International Conference on Neural Information Processing Systems},
  pages={20645--20693},
  year={2024}
}

@incollection{bogdanov2017pseudorandom,
  title={Pseudorandom functions: Three decades later},
  author={Bogdanov, Andrej and Rosen, Alon},
  booktitle={Tutorials on the Foundations of Cryptography: Dedicated to Oded Goldreich},
  pages={79--158},
  year={2017},
  publisher={Springer}
}

@article{mathur2007design,
  title={On the design of error-correcting ciphers},
  author={Mathur, Chetan Nanjunda and Narayan, Karthik and Subbalakshmi, KP},
  journal={EURASIP Journal on Wireless Communications and Networking},
  volume={2006},
  number={1},
  pages={042871},
  year={2007},
  publisher={Springer}
}

@article{aragon2022bike,
  title={BIKE: bit flipping key encapsulation},
  author={Aragon, Nicolas and Barreto, Paulo and Bettaieb, Slim and Bidoux, Lo{\"\i}c and Blazy, Olivier and Deneuville, Jean-Christophe and Gaborit, Philippe and Ghosh, Santosh and Gueron, Shay and G{\"u}neysu, Tim and others},
  year={2022}
}

@inproceedings{bos2018crystals,
  title={CRYSTALS-Kyber: a CCA-secure module-lattice-based KEM},
  author={Bos, Joppe and Ducas, L{\'e}o and Kiltz, Eike and Lepoint, Tancr{\`e}de and Lyubashevsky, Vadim and Schanck, John M and Schwabe, Peter and Seiler, Gregor and Stehl{\'e}, Damien},
  booktitle={2018 IEEE European symposium on security and privacy (EuroS\&P)},
  pages={353--367},
  year={2018},
  organization={IEEE}
}

@misc{c2pa,
  author = {{Coalition for Content Provenance and Authenticity}},
  title  = {C2PA},
  url    = {https://c2pa.org/}
}

@misc{openai_c2pa,
  author = {{OpenAI}},
  title  = {C2PA and SynthID in OpenAI-Generated Images},
  url    = {https://help.openai.com/en/articles/8912793-c2pa-in-images}
}

@misc{azure_content_credentials,
  author = {{Microsoft}},
  title  = {Content Credentials in Azure OpenAI},
  url    = {https://learn.microsoft.com/en-us/azure/ai-services/openai/concepts/content-credentials}
}

@misc{claude,
    author = {{Claude}},
    title = {How Claude marks AI-generated content},
    url = {https://support.claude.com/}
}

@misc{google_synthid,
  author = {{Google DeepMind}},
  title  = {SynthID: Identifying AI-Generated Content},
  url    = {https://deepmind.google/technologies/synthid/}
}

@article{hu2025videomark,
  title={Videomark: A distortion-free robust watermarking framework for video diffusion models},
  author={Hu, Xuming and Li, Hanqian and Li, Jungang and Huang, Yu and Liu, Shuliang and Zheng, Qi and Chen, Junhao and Liu, Aiwei},
  journal={arXiv preprint arXiv:2504.16359},
  year={2025}
}

\appendix


\section{Open Science}
All authors have read and consented to the open science policy. We are committed to releasing our artifacts to the community to ensure reproducibility and to benefit the broader research ecosystem. The artifacts are anonymously available at:
\url{https://anonymous.4open.science/r/USENIX-CECMark-6E87/}.

\mypara{Code Availability}
The artifact includes implementations for:
\begin{itemize}
\item CECMark and the underlying Conc-CEC scheme;
\item Baseline watermark schemes, including GS, PRCMark, and T2SMark;
\item Evaluation metrics, such as CLIP, FID, decoding correctness, latent sign flipping, and distinguishability;
\item  Statistical analysis and figure generation.
\end{itemize}
All code is clearly documented in the README file, and the project structure is carefully organized to ensure clarity, reproducibility, and extensibility for future research.

\mypara{Data Availability}
Due to the large size of the generated datasets, we do not include the full set of images and corresponding latents in the artifact.
Instead, we provide a one‑request demo dataset for each scheme.
We also provide the statistical results over the full dataset to reproduce the figures in the paper.
Public resources, such as the Stable Diffusion Prompt dataset and the Stable Diffusion 2.1 Base model, are excluded from the artifact.

\section{Ethical Considerations}
\label{sect:ethical}
This work introduces a new cryptographic primitive named CEC, and applies it to design CECMark for watermarking AI-generated images of generative image models.
While CECMark improves the confidentiality, robustness and soundness of the watermark, several ethical considerations remain to be discussed.

\mypara{Privacy Risks of Embedded Information}
The watermarking techniques might be used to embed information such as user identifiers or provenance data, which introduces potential privacy risks if used without proper authorization.
Although CECMark provides confidentiality against adversaries without access to the secret key, these risks stem from how watermarking is deployed rather than the embedding itself.
To mitigate these concerns, users should be clearly informed and give explicit consent before their information is embedded into watermarks, and such mechanisms should only be used for legitimate purposes such as provenance or attribution.

\mypara{Security of Pseudorandom Function Instantiation}
As shown in our work, the confidentiality of CECMark is reduced to the security of pseudorandom functions.
When deploying watermarking systems in practice, we suggest that cryptographically secure pseudorandom functions should be used.
If model providers replace pseudorandom functions with non-cryptographic random generators for simplicity, it may violate the underlying security assumptions and lead to potential information leakage.

\mypara{Impact Control of Theoretical Attacks}
Although we reveal the limitations of existing watermarking schemes including GS, PRCMark, and T2SMark, our goal is to avoid the deployment of schemes with potential security risks. 
In our discussions, the confidentiality of these schemes is not satisfied under practical parameter sets. 
Since these watermarking schemes are not yet deployed in real-world environments, our theoretical attacks did not cause any real-world privacy leaks or other negative effects. 
As a result, we propose CEC with provable confidentiality to prevent potential attacks and support the secure deployment of watermarking schemes.

\mypara{Responsible Deployment and Usage}
We emphasize that watermarking should only be applied in appropriate application scenarios with clearly defined purposes.
System providers should carefully evaluate the impact of watermarking on system behavior and ensure that its deployment aligns with the intended functionality of generative models.
More broadly, watermarking is not a complete solution to the misuse of generative models, and should be combined with other mechanisms for content governance and authenticity.

\section{Related Works}
\label{sect:related_works}
\subsection{Error-Tolerant Cryptography}
\label{sec:relatedwork_error}
The cryptographic schemes with error-tolerant capability are mainly divided into two categories.
\textbf{(1) Public-key cryptography} based on codes and lattices inherently involves noise during encryption, such as BIKE~\cite{aragon2022bike}, HQC~\cite{gaborit2025hamming}, and Kyber~\cite{bos2018crystals}.
However, the injected noise in these schemes is primarily used to ensure security, rather than to provide robustness against adversarial perturbations.
Moreover, the magnitude and distribution of the noise are tightly coupled with the underlying algebraic structures of codes or lattices, and must be carefully controlled to guarantee correct decryption.
As a result, the tolerable error rate is inherently limited.
For example, the error rates in BIKE and HQC are around $0.5\%$ and $0.4\%$, respectively, whereas the error rate under watermark removal attacks is on the order of $10\%$, which is significantly larger.
Therefore, such schemes offer limited robustness in applications like watermarking.
(2) \textbf{Secret-key cryptography} can also be combined with error-correcting codes.
~\cite{mathur2007design} combines AES encryption with Reed-Solomon codes to achieve secure and reliable communication over wireless channels.
However, their construction mainly focuses on integrating error-correcting capability with a deterministic block cipher, without explicitly considering the randomized encryption mechanisms required for IND-CPA security.
In particular, modern symmetric encryption schemes typically rely on a fresh nonce or initial vector for each encryption, which is treated as part of the ciphertext.
Such randomness is not incorporated into the design or analysis in~\cite{mathur2007design}.
Therefore, their construction does not capture the setting considered in our work, where both randomness and error tolerance must be jointly handled within the encryption process.

\subsection{Post-processing Watermarks}
Post-processing schemes embed watermarks into images after they are generated by GIMs~\cite{ShieldMntInvisibleWatermark,Tancik_2020_CVPR,gowal2025synthidimageimagewatermarkinginternet}.
Such schemes are generally implemented using one of two approaches: steganography‑based or encoder-decoder-based~\cite{pmlr-v235-an24a}.
DwtDct~\cite{ShieldMntInvisibleWatermark} is one of the most widely used steganography-based methods, which combines the Discrete Wavelet Transform (DWT) and the Discrete Cosine Transform (DCT) to encode the message into the frequency domain of the target image.
However, this technique often introduces visible artifacts and lacks provable robustness guarantees.
The encoder-decoder-based approach trains an encoder for watermark embedding and a decoder for watermark detection~\cite{Tancik_2020_CVPR,gowal2025synthidimageimagewatermarkinginternet}.
A representative example is SynthID-Image~\cite{gowal2025synthidimageimagewatermarkinginternet}, which is a production-level private watermark scheme developed by Google DeepMind.
While this paradigm can achieve a practical balance between image quality and watermark robustness, it lacks provable security and soundness guarantees.
Consequently, it remains susceptible to black‑box watermark removal attacks\footnote{\url{https://synthidbypass.com/}.}.

\section{Encoding Procedure of Concatenated Codes}
\label{sect:encode_concatenated}

The encoding procedure $\vec{v}=\mathcal{C}.\mathsf{Encode}(\vec{m})$ of the concatenated code is shown in Figure~\ref{fig:concatenated_code} and the decoding procedure follows analogously.

Let $q = 2^{k_i}$ and $f: \GF{q} \rightarrow \GF{2}^{k_i}$
be a fixed bijection that maps each element in $\GF{q}$ to its
$k_i$-bit representation over $\GF{2}$.
The message space of the concatenated code is $\GF{2}^{k_e k_i}$.
Given a binary message
$\vec{m} \in \GF{2}^{k_e k_i}$,
we partition it into $k_e$ blocks
$\vec{m} = (\vec{m}_1,\ldots,\vec{m}_{k_e})$, where $\vec{m}_j \in \GF{2}^{k_i}$.
Each block is mapped to a symbol in $\GF{q}$ via $f^{-1}$, yielding
\[
\vec{m}^{(q)} = \bigl(f^{-1}(\vec{m}_1),\ldots,f^{-1}(\vec{m}_{k_e})\bigr)
\in \GF{q}^{k_e}.
\]
The external code encodes $\vec{m}^{(q)}$ as
\[
\vec{c}^{(q)} = \mathcal{C}_e.\mathsf{Encode}(\vec{m}^{(q)})
= (c_1^{(q)},\ldots,c_{n_e}^{(q)}),
\qquad c_j^{(q)} \in \GF{q}.
\]
Next, each symbol $c_j^{(q)}$ is first mapped to $\GF{2}^{k_i}$ via $f$ to obtain
\[
\vec{c} =  (\vec{c}_1,\ldots,\vec{c}_{n_e}), \qquad 
\vec{c}_j = f(\vec{c}_j^{(q)}) \in \GF{2}^{k_i},
\]
and then encoded independently using the internal code,
\[
\vec{v}_j = \mathcal{C}_i.\mathsf{Encode}(\vec{c}_j) \in \GF{2}^{n_i}.
\]
Finally, the overall binary codeword is obtained by concatenation:
\[
\vec{v} = \mathcal{C}.\mathsf{Encode}(\vec{m})
= (\vec{v}_1, \ldots, \vec{v}_{n_e})
\in \GF{2}^{n_en_i}.
\]

\section{Games for the IND-CPA Proof}
\label{sect:ind_cpa_games}
This section presents the modified games used in the proof of Theorem~\ref{thm:ind_cpa}.

\refstepcounter{game}
\begin{tcolorbox}[
breakable,
enhanced jigsaw,
colback=white,
colframe=black,
coltitle=black,
colbacktitle=white,
title={Game~\thegame~$\mathsf{Game}_{f_1,\mathcal{B}_1}^{\mathsf{PRF}}(\lambda)$ (Distinguish $f_1$ from random)},
fonttitle=\bfseries,
boxrule=0.5pt,
boxsep=0pt,
top=0.5pt,
bottom=0.5pt,
left skip=3pt,
right skip=3pt,
]
\label{game:prf-b1}

\begin{enumerate}[leftmargin=1.2em,label=\arabic*.,nosep]
    \item $b \xleftarrow{\$} \{0,1\}$
    
    \item $\mathcal{O}_1(\cdot) :=
    \begin{cases}
        f_1(\vec{mk}_1) & \text{if } b = 0 \\
        R_1(\cdot)       & \text{if } b = 1
    \end{cases}$
    
    \item $\vec{mk}_2 \xleftarrow{\$} \GF{2}^{\lambda_1}$
    
    \item $\mathcal{O}(\vec{m})$:
    \begin{enumerate}
        \item $\vec{\theta} \xleftarrow{\$} \GF{2}^{\lambda_2}$
        \item $\vec{c}_1 \leftarrow \Encode_1(\vec{\theta}) + \mathcal{O}_1()$
        \item $\vec{c}_2 \leftarrow \Encode_2(\vec{m}) + f_2(\vec{mk}_2,\vec{\theta})$
        \item return $(\vec{c}_1,\vec{c}_2)$
    \end{enumerate}
    
    \item $(\vec{m}_0,\vec{m}_1) \leftarrow 
    \mathcal{A}_{\mathsf{CHOOSE}}^{\mathcal{O}}(1^\lambda)$
    
    \item $\beta \xleftarrow{\$} \{0,1\}$
    
    \item ${\theta}^* \xleftarrow{\$} \GF{2}^{\lambda_2}$
    
    \item $\vec{c}_1^* \leftarrow \Encode_1({\theta}^*) + \mathcal{O}_1()$
    
    \item $\vec{c}_2^* \leftarrow \Encode_2(\vec{m}_{\beta}) + f_2(\vec{mk}_2,{\theta}^*)$
    
    \item $\hat{\beta} \leftarrow 
    \mathcal{A}_{\mathsf{GUESS}}^{\mathcal{O}}(1^\lambda,(\vec{c}_1^*,\vec{c}_2^*))$

    \item $\mathcal{B}_1$ outputs $b' \leftarrow 1-(\hat{\beta} = \beta)$
    
    \item return $(b' = b)$
\end{enumerate}

\end{tcolorbox}

\refstepcounter{game}
\begin{tcolorbox}[
breakable,
enhanced jigsaw,
colback=white,
colframe=black,
coltitle=black,
colbacktitle=white,
title={Game~\thegame~$\mathsf{Game}_{f_2,\mathcal{B}_2}^{\mathsf{PRF}}(\lambda)$
 (Distinguish $f_2$ from random)},
fonttitle=\bfseries,
boxrule=0.5pt,
boxsep=0pt,
top=0.5pt,
bottom=0.5pt,
left skip=3pt,
right skip=3pt,
]
\label{game:prf-b2}

\begin{enumerate}[leftmargin=1.2em,label=\arabic*.,nosep]
    \item $b \xleftarrow{\$} \{0,1\}$
    
    \item $\mathcal{O}_2(\cdot) :=
    \begin{cases}
        f_2(\vec{mk}_2,\cdot) & \text{if } b = 0 \\
        R_2(\cdot)            & \text{if } b = 1
    \end{cases}$
    
    \item $\vec{r}_1 \xleftarrow{\$} \GF{2}^{n_1}$
    
    \item $\mathcal{D} \leftarrow \emptyset$
    
    \item $\mathcal{O}(\vec{m})$:
    \begin{enumerate}
        \item $\theta \xleftarrow{\$} \GF{2}^{\lambda_2}$
        \item $\vec{c}_1 \leftarrow \Encode_1(\theta) + \vec{r}_1$
        \item $\text{if } \theta\notin\mathcal{D}:\;
        \mathcal{D}[\theta] \leftarrow \mathcal{O}_2(\theta)$
        \item $\vec{s} \leftarrow \mathcal{D}[\theta]$
        \item $\vec{c}_2 \leftarrow \Encode_2(\vec{m}) + \vec{s}$
        \item return $(\vec{c}_1,\vec{c}_2)$
    \end{enumerate}
    
    \item $(\vec{m}_0,\vec{m}_1) \leftarrow 
    \mathcal{A}_{\mathsf{CHOOSE}}^{\mathcal{O}}(1^\lambda)$
    
    \item $\beta \xleftarrow{\$} \{0,1\}$
    
    \item $\theta^* \xleftarrow{\$} \GF{2}^{\lambda_2}$
    
    \item $\vec{c}_1^* \leftarrow \Encode_1(\theta^*) + \vec{r}_1$
    
    \item $\text{if } \theta^*\notin\mathcal{D}:\;
    \mathcal{D}[\theta^*] \leftarrow \mathcal{O}_2(\theta^*)$
    
    \item $\vec{s}^* \leftarrow \mathcal{D}[\theta^*]$
    
    \item $\vec{c}_2^* \leftarrow \Encode_2(\vec{m}_\beta) + \vec{s}^*$
    
    \item $\hat{\beta} \leftarrow 
    \mathcal{A}_{\mathsf{GUESS}}^{\mathcal{O}}(1^\lambda,(\vec{c}_1^*,\vec{c}_2^*))$
    
    \item $\mathcal{B}_2$ outputs $b' \leftarrow 1-(\hat{\beta} = \beta)$
    
    \item return $(b' = b)$
\end{enumerate}

\end{tcolorbox}

\refstepcounter{game}
\begin{tcolorbox}[
breakable,
enhanced jigsaw,
colback=white,
colframe=black,
coltitle=black,
colbacktitle=white,
title={Game~\thegame~$\mathsf{Game}_{\mathcal{A}}^{\mathsf{IND\text{-}CPA}}(\lambda)$ (Replace $f_1$ with randomness)},
fonttitle=\bfseries,
boxrule=0.5pt,
boxsep=0pt,
top=0.5pt,
bottom=0.5pt,
left skip=3pt,
right skip=3pt,
]
\label{game:cec-game1}

\begin{enumerate}[leftmargin=1.2em,label=\arabic*.,nosep]
    \item \textcolor{blue}{$\vec{mk}_2 \xleftarrow{\$} \GF{2}^{\lambda_1}$}
    
    \item \textcolor{blue}{$\vec{r}_1 \xleftarrow{\$} \GF{2}^{n_1}$}
    
    \item $\mathcal{O}(\vec{m})$:
    \begin{enumerate}
        \item $\vec{\theta} \xleftarrow{\$} \GF{2}^{\lambda_2}$
        \item \textcolor{blue}{$\vec{c}_1 \leftarrow \Encode_1(\vec{\theta}) + \vec{r}_1$}
        \item $\vec{c}_2 \leftarrow \Encode_2(\vec{m}) + f_2(\vec{mk}_2,\vec{\theta})$
        \item return $(\vec{c}_1,\vec{c}_2)$
    \end{enumerate}
    
    \item $(\vec{m}_0,\vec{m}_1) \leftarrow 
    \mathcal{A}_{\mathsf{CHOOSE}}^{\mathcal{O}}(1^\lambda)$
    
    \item $b \xleftarrow{\$} \{0,1\}$
    
    \item $\vec{\theta}^* \xleftarrow{\$} \GF{2}^{\lambda_2}$
    
    \item \textcolor{blue}{$\vec{c}_1^* \leftarrow \Encode_1(\vec{\theta}^*) + \vec{r}_1$}
    
    \item $\vec{c}_2^* \leftarrow \Encode_2(\vec{m}_b) + f_2(\vec{mk}_2,\vec{\theta}^*)$
    
    \item $b' \leftarrow 
    \mathcal{A}_{\mathsf{GUESS}}^{\mathcal{O}}(1^\lambda,(\vec{c}_1^*,\vec{c}_2^*))$
    
    \item return $(b' = b)$
\end{enumerate}

\end{tcolorbox}

\refstepcounter{game}
\begin{tcolorbox}[
breakable,
enhanced jigsaw,
colback=white,
colframe=black,
coltitle=black,
colbacktitle=white,
title={Game~\thegame~$\mathsf{Game}_{\mathcal{A}}^{\mathsf{IND\text{-}CPA}}(\lambda)$ (Replace $f_2$ with randomness)},
fonttitle=\bfseries,
boxrule=0.5pt,
boxsep=0pt,
top=0.5pt,
bottom=0.5pt,
left skip=3pt,
right skip=3pt,
]
\label{game:cec-game2}
\begin{enumerate}[leftmargin=1.2em,label=\arabic*.,nosep]
    \item \textcolor{red}{\sout{$\vec{mk}_2 \xleftarrow{\$} \GF{2}^{\lambda_1}$}}
    
    \item $\vec{r}_1 \xleftarrow{\$} \GF{2}^{n_1}$
    
    \item \textcolor{blue}{$\mathcal{D} \leftarrow \emptyset$}
    
    \item $\mathcal{O}(\vec{m})$:
    \begin{enumerate}
        \item ${\theta} \xleftarrow{\$} \GF{2}^{\lambda_2}$
        \item $\vec{c}_1 \leftarrow \Encode_1({\theta}) + \vec{r}_1$
        \item \textcolor{blue}{
$\text{if } {\theta}\notin\mathcal{D}:\;
\mathcal{D}[\vec{\theta}] \xleftarrow{\$} \GF{2}^{n_2}$}
        \item \textcolor{blue}{$\vec{s} \leftarrow \mathcal{D}[\vec{\theta}]$}
        \item $ \textcolor{blue}{\vec{c}_2 \leftarrow \Encode_2(\vec{m}) +\vec{s}}$
        \item return $(\vec{c}_1,\vec{c}_2)$
    \end{enumerate}
    
    \item $(\vec{m}_0,\vec{m}_1) \leftarrow 
    \mathcal{A}_{\mathsf{CHOOSE}}^{\mathcal{O}}(1^\lambda)$
    
    \item $b \xleftarrow{\$} \{0,1\}$
    
    \item ${\theta}^* \xleftarrow{\$} \GF{2}^{\lambda_2}$
    
    \item $\vec{c}_1^* \leftarrow \Encode_1({\theta}^*) + \vec{r}_1$
    
   \item \textcolor{blue}{
$\text{if } \vec{\theta^*}\notin\mathcal{D}:\;
\mathcal{D}[{\theta^*}] \xleftarrow{\$} \GF{2}^{n_2}$}
        \item \textcolor{blue}{$\vec{s}^* \leftarrow \mathcal{D}[{\theta}^*]$}
    
    \item $\textcolor{blue}{\vec{c}_2^* \leftarrow \Encode_2(\vec{m}_b) + \vec{s}^*}$
    
    \item $b' \leftarrow 
    \mathcal{A}_{\mathsf{GUESS}}^{\mathcal{O}}(1^\lambda,(\vec{c}_1^*,\vec{c}_2^*))$
    
    \item return $(b' = b)$
\end{enumerate}

\end{tcolorbox}

\begin{figure}[]
\centering
\includegraphics[width=0.9\linewidth]{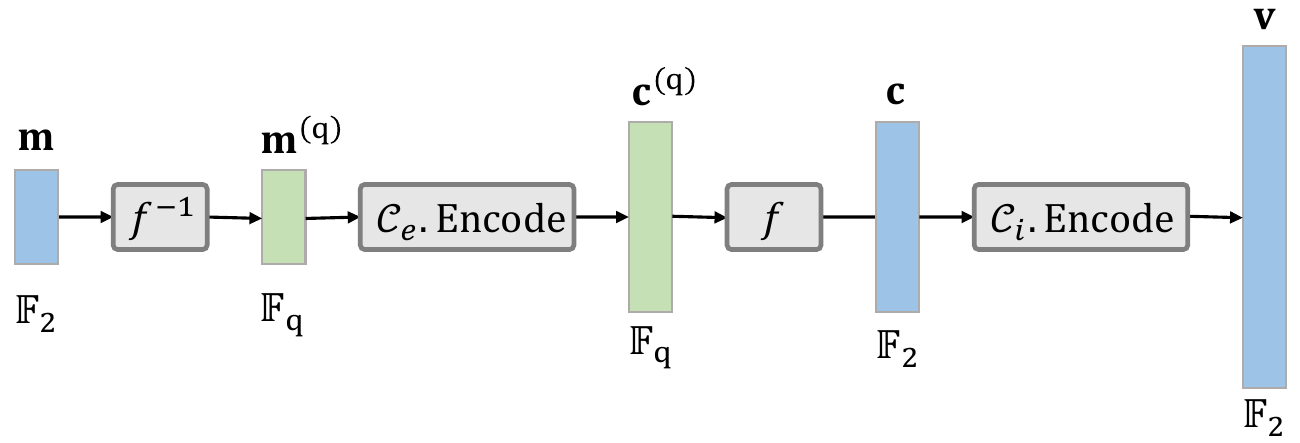}
\caption{Encoding procedure of the concatenated codes.}
\label{fig:concatenated_code}
\end{figure}

\section{Details of Baseline Watermarking Schemes}
\label{sect:baseline-details}
For GS, the payload is encoded using a duplicated code with $(1,8,8)$ repetitions along the channel, width, and height dimensions, respectively.
We use the same-key-same-nonce setting across all generations to simulate a real-world key management scenario~\cite{thietke2025towards}.
For PRCMark, the payload is encoded using the LDPC-PRC with parameter $t=3$ and a theoretical false positive rate of $10^{-5}$.
Decoding is performed using the soft-decoding technique~\cite{DBLP:conf/iclr/GunnZS25} with $\sigma=1.5$.
For T2SMark, a 16-bit session key is embedded in the first channel, while the payload is embedded across the remaining three channels. 
Each component is encoded with a duplicated code, ensuring 50\% occupancy within its designated channels.
Note that the detection and decoding mechanisms of T2SMark are independent.
We evaluate its robustness by directly decoding the message without reference to the detection result, and evaluate its soundness by empirically setting the detection threshold to 300, which produces zero false positives under incorrect keys.
These settings advantage T2SMark, yielding an overestimated upper bound in these metrics.

\section{Details of Watermark Robustness}
\label{sect:pert}

\mypara{Implementation of Image Transformations}
\begin{figure}[tbp]
\centering
\includegraphics[width=1\linewidth]{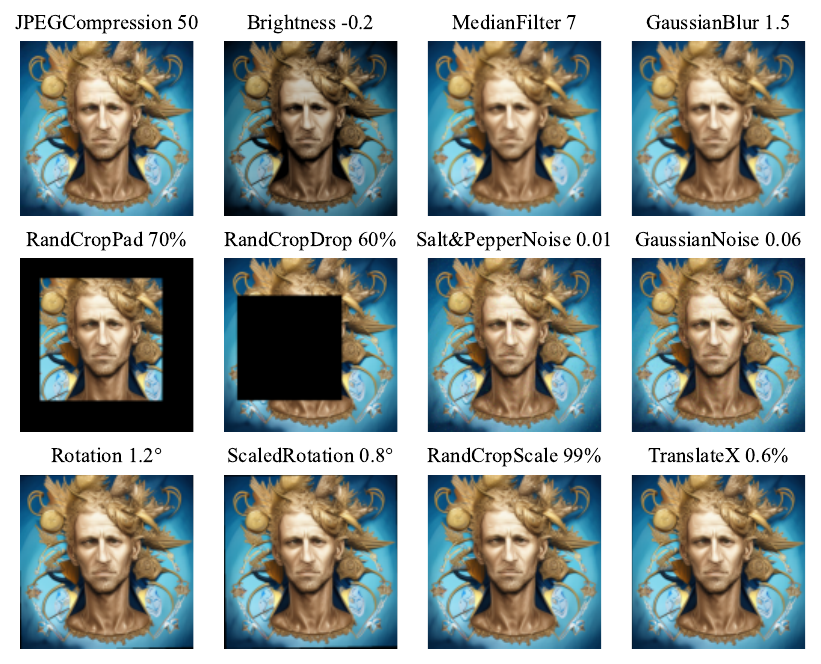}
\caption{Example images under different perturbations.}
\label{fig:pert-image-examples}
\end{figure}
We adopt 12 image transformations as shown in~\Cref{fig:pert-image-examples}, which can be roughly divided into four categories: global transformations, masking, noise, and geometric transformations. 
Their implementations are detailed as follows:

\begin{figure*}[t]
    \centering
    \includegraphics[width=1\linewidth]{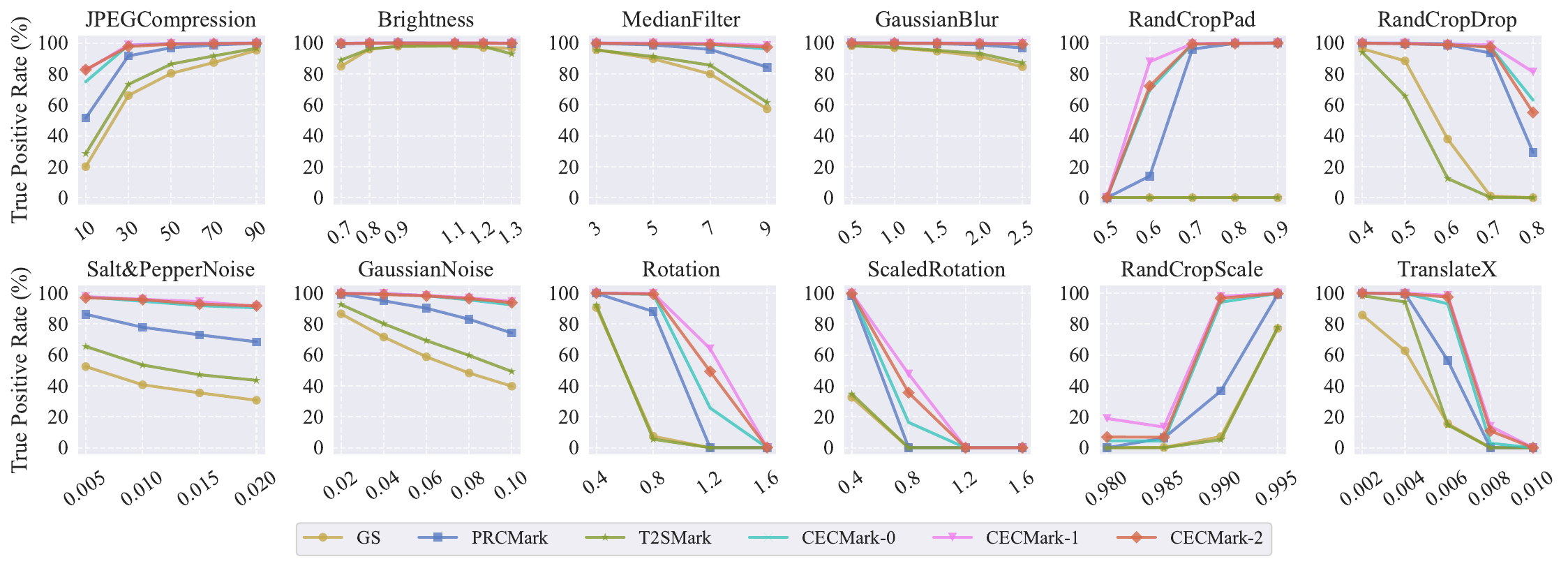}
    \caption{The True Positive Rate (TPR) of various watermark schemes under transformations of different strengths.}
    \label{fig:pert-image-tpr}
\end{figure*}

\begin{figure}[htb]
\centering
\includegraphics[width=0.9\linewidth]{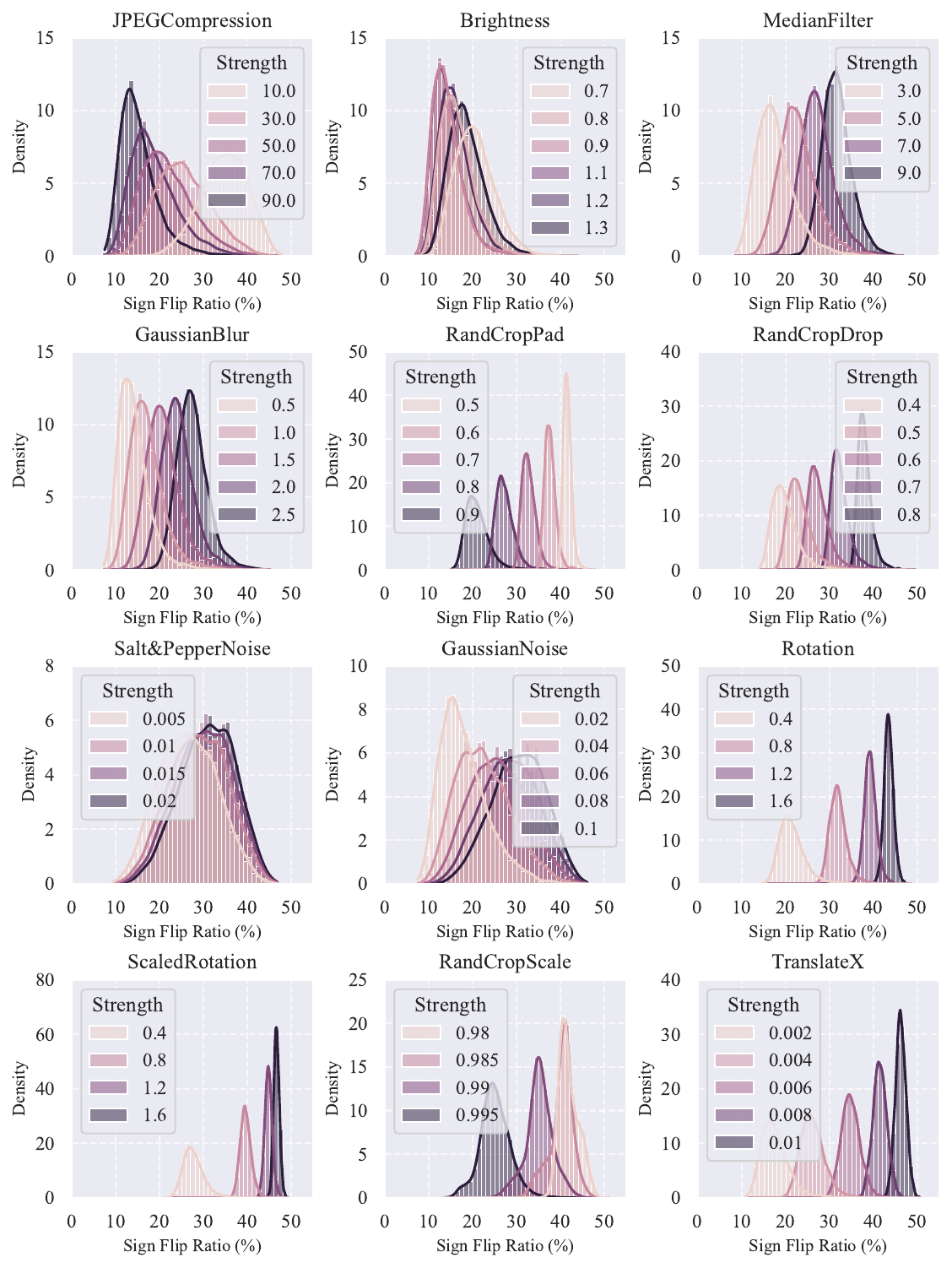}
\caption{Ratio of latent sign flips induced by different image perturbations.}

\label{fig:pert-sign-flip}
\end{figure}

\begin{itemize}
    \item \textbf{Global Transformations}:
\begin{itemize}
    \item \textbf{JPEGCompression}: Apply JPEG compression with quality factor $\alpha$, where a higher $\alpha$ indicates better quality preservation.
    \item \textbf{Brightness}: Adjust image brightness globally. Brighter for $\alpha>1$ and darker for $\alpha<1$.
    \item \textbf{MedianFilter}: Apply a median filter with kernel size $\alpha$, where a larger $\alpha$ results in stronger smoothing and loss of fine details.
    \item \textbf{GaussianBlur}: Apply Gaussian blur with standard deviation $\alpha$, where a larger $\alpha$ produces stronger blurring effects.
\end{itemize}
\item \textbf{Masking}:
\begin{itemize}
    \item \textbf{RandCropPad}: Randomly crop a square covering $\alpha\%$ of pixels, retain the cropped region, and pad the remaining area with zeros. A larger $\alpha$ preserves more of the original image.
    \item \textbf{RandCropDrop}: Randomly crop a square covering $\alpha\%$ of pixels, and set the cropped region to zeros. A larger $\alpha$ drops more content from the image.
\end{itemize}
\item \textbf{Noise}:
\begin{itemize}
    \item \textbf{Salt\&PepperNoise}: Randomly replace $\alpha(\%)$ of pixels with either maximum or minimum intensity values. A larger $\alpha$ introduces more impulsive noise.
    \item \textbf{GaussianNoise}: Add Gaussian noise with variance $\alpha$ to the image. A larger $\alpha$ corresponds to stronger noise levels.
\end{itemize}
\item \textbf{Geometric Transformations}:
\begin{itemize}
    \item \textbf{Rotation}: Rotate the image by $\alpha^\circ$, discarding pixels that fall outside the bounding box.
    \item \textbf{ScaledRotation}: Rotate the image by $\alpha^\circ$ and scale it down to fit the bounding box.
    \item \textbf{RandCropScale}: Randomly crop a square covering $\alpha\%$ of pixels, then scale the cropped region to the full image size.
    \item \textbf{TranslateX}: Translate the image horizontally by $\alpha\%$ of pixels.
\end{itemize}
\end{itemize}

\mypara{Per-perturbation Results}
The per-perturbation TPRs of each watermark scheme are presented in~\Cref{fig:pert-image-tpr}. Across all perturbations, CECMark consistently achieves higher robustness than existing watermark schemes, showing its strong robustness.
Nevertheless, it still fails under certain transformations.
To analyze these failures, we examine the ratios of latent sign flips induced by different perturbations in~\Cref{fig:pert-sign-flip}.
The results indicate that these perturbations generate nearly random latent flips, which exceed the decoding boundary of CECMark.
In fact, such strong perturbations invalidate nearly any error-correcting scheme~\cite{cryptoeprint:2025/1620prclimits}, representing an inherent limitation of latent-based watermarking that CECMark cannot fully resolve.


\section{Efficiency of Different Watermarking Schemes}
\label{sec:throughput}
\begin{table}[t]
\centering
\caption{The throughput of different watermark schemes for watermark encoding and decoding (iterations per second, it/s).}
\label{tab:speed}
\setlength{\tabcolsep}{10pt}
\begin{tabular}{crr}
\toprule
Schemes & Encode ($\uparrow$) & Decode ($\uparrow$) \\ \midrule
GS & 254.85 & 998.14 \\
PRCMark & 3.40 & 3.58 \\
T2SMark & 483.32 & 839.08 \\ \midrule
RSSMark-0 & 250.16 & 7.79 \\
RSSMark-1 & 243.37 & 10.46 \\
RSSMark-2 & 247.02 & 10.18 \\ \bottomrule
\end{tabular}
\end{table}

The encoding and decoding throughputs (iterations per second, it/s) of different watermarking schemes are reported in Table~\ref{tab:speed}.
The evaluation is conducted by performing 1,000 encoding and decoding operations implemented in Python on a single process of an Intel Xeon Platinum 8558P processor.

For encoding, RSSMark achieves 243--250 it/s, with only marginal differences across variants.
The throughput of RSSMark is comparable to GS, which only requires ChaCha20 bitstream generation.
Although RSSMark additionally performs RM-RS encoding, this operation is linear-time and only introduces negligible overhead.
T2SMark achieves the highest throughput because it uses Python's built-in pseudorandom number generator instead of the cryptographically secure ChaCha20 stream.
In contrast, PRCMark achieves only 3.40 it/s due to its expensive matrix operations.

For decoding, RSSMark achieves a low throughput of 7.79--10.46 it/s.
This is because RSSMark performs computationally intensive RM-RS decoding.
By comparison, GS and T2SMark employ much simpler decoding procedures and therefore achieve significantly higher throughput.
Nevertheless, RSSMark still outperforms PRCMark, whose Belief Propagation decoder incurs even higher computational complexity.

Although CECMark exhibits asymmetric encoding and decoding throughput, this is practical because watermark encoding is performed far more frequently than watermark decoding in real-world deployments. Moreover, compared with image generation and latent inversion, which achieve only 0.56 it/s and 0.85 it/s on a single NVIDIA H20 GPU, respectively, watermark encoding and decoding are not the bottleneck of the overall pipeline.

\section{Illustration of Watermark Embedding and Extraction}
\label{sect:cecmark_embed_extract}
The overall process from the initial latent to the final image is shown in \Cref{fig:watermark_flow}. We further illustrate the specific embedding and extraction procedures between the message and the watermarked latent in the following discussion.

\mypara{Embedding}
In Figure~\ref{fig:watermark_embedding}, the latent vector is partitioned according to the magnitudes of its coordinates. 
The green and orange regions correspond to the selected coordinates used for watermark embedding, while the blue region contains the remaining coordinates that are not modified. 
The selected coordinates are mapped to the first part of the channel, where their signs are adjusted according to the watermark bits. 
Afterward, a global permutation is applied to obtain the final watermarked latent.

\mypara{Extraction}
Figure~\ref{fig:watermark_extraction} illustrates the reverse process. 
The latent is first inverse-permuted to restore the canonical ordering. 
The watermark bits are then extracted from the signs of the coordinates in the first part of the channel.


\begin{figure}[]
\centering
\includegraphics[width=1\linewidth]{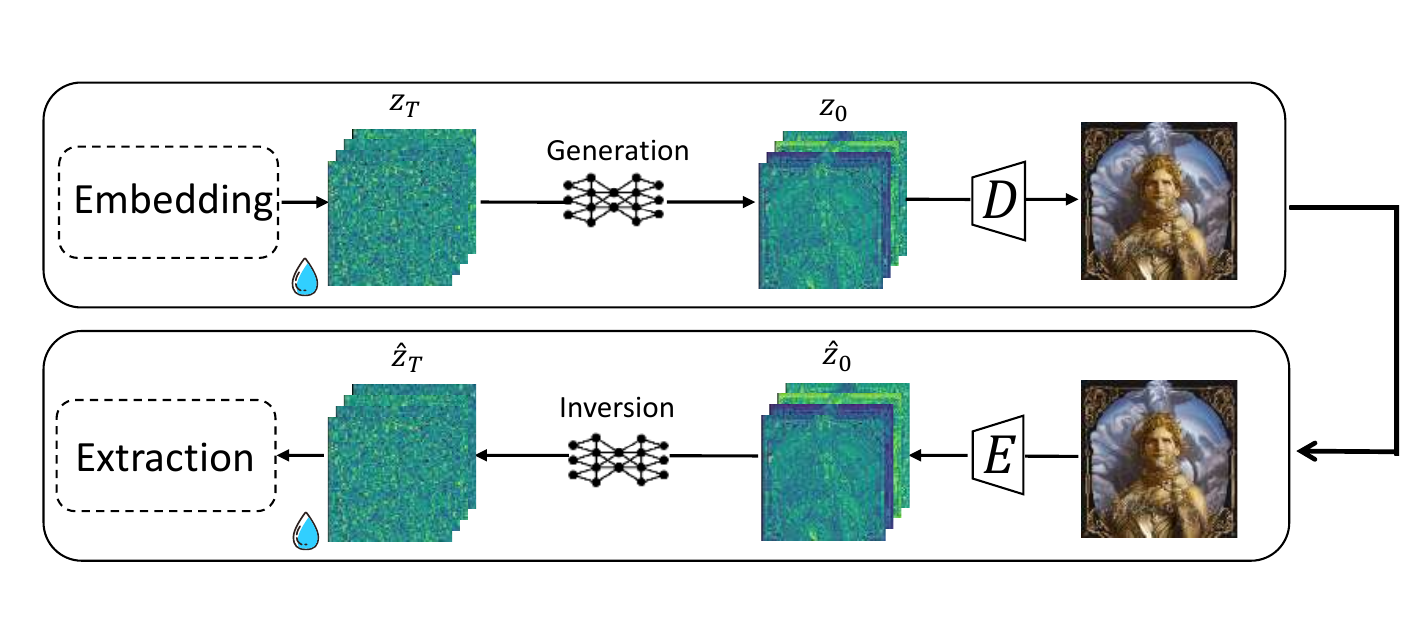}
\caption{Overview of the generation and inversion processes between the initial latent and the image. The latent consists of four channels. Generation and inversion are implemented using \texttt{U-Net}, while \texttt{D-Net} ($\mathcal{D}$) and \texttt{E-Net} ($\mathcal{E}$) denote the decoder and encoder networks, respectively.}
\label{fig:watermark_flow}
\end{figure}

\begin{figure}[t]
  \centering
  \includegraphics[width=0.85\linewidth]{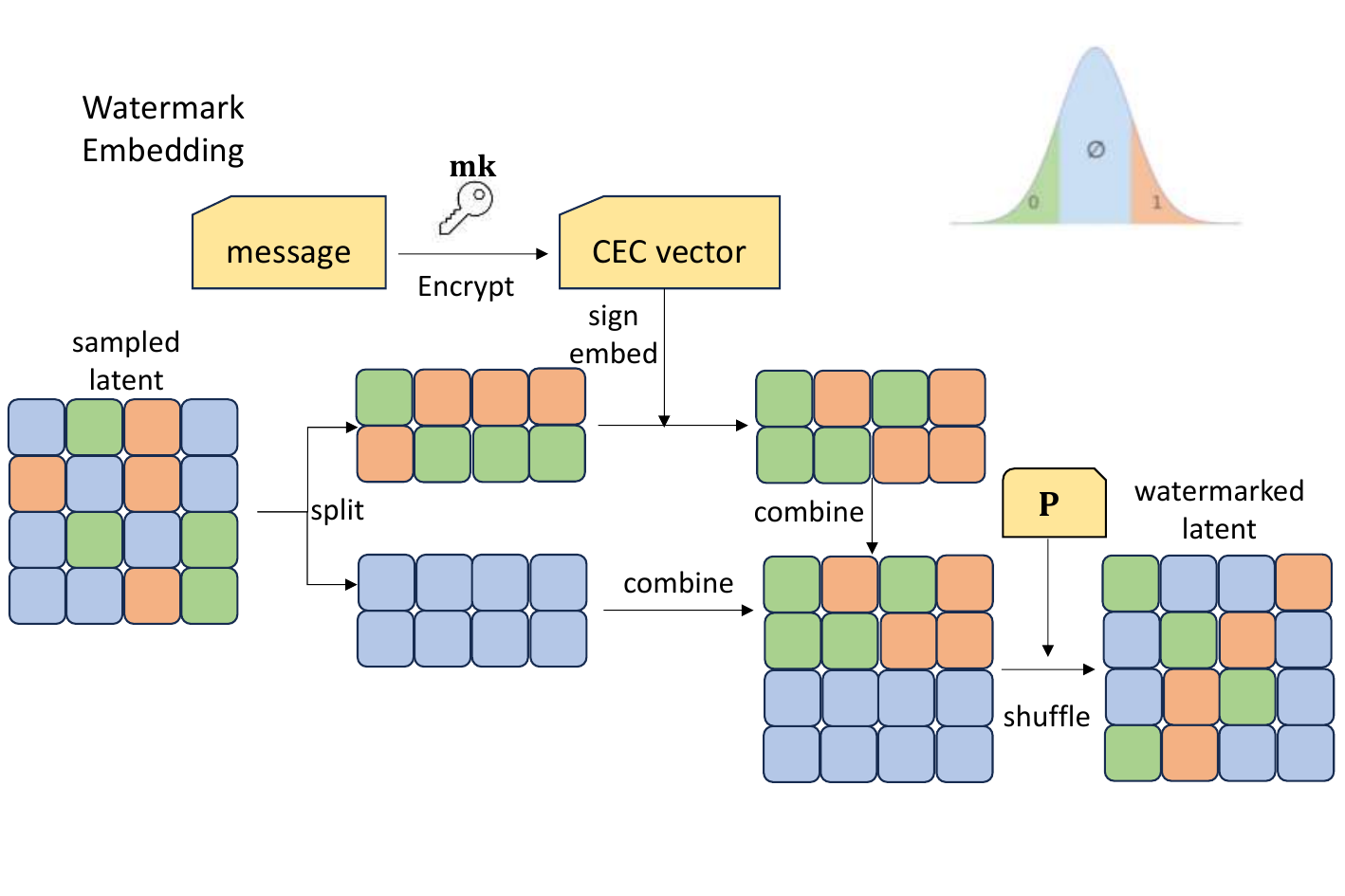}
  \caption{Watermark embedding of CECMark. The green, blue, and orange areas represent positions with large negative, near-zero, and large positive values, respectively. The watermark is only embedded into regions with high magnitude, i.e., the green and orange areas.
  }
  \label{fig:watermark_embedding}
\end{figure}

\begin{figure}[t]
  \centering
  \includegraphics[width=0.85\linewidth]{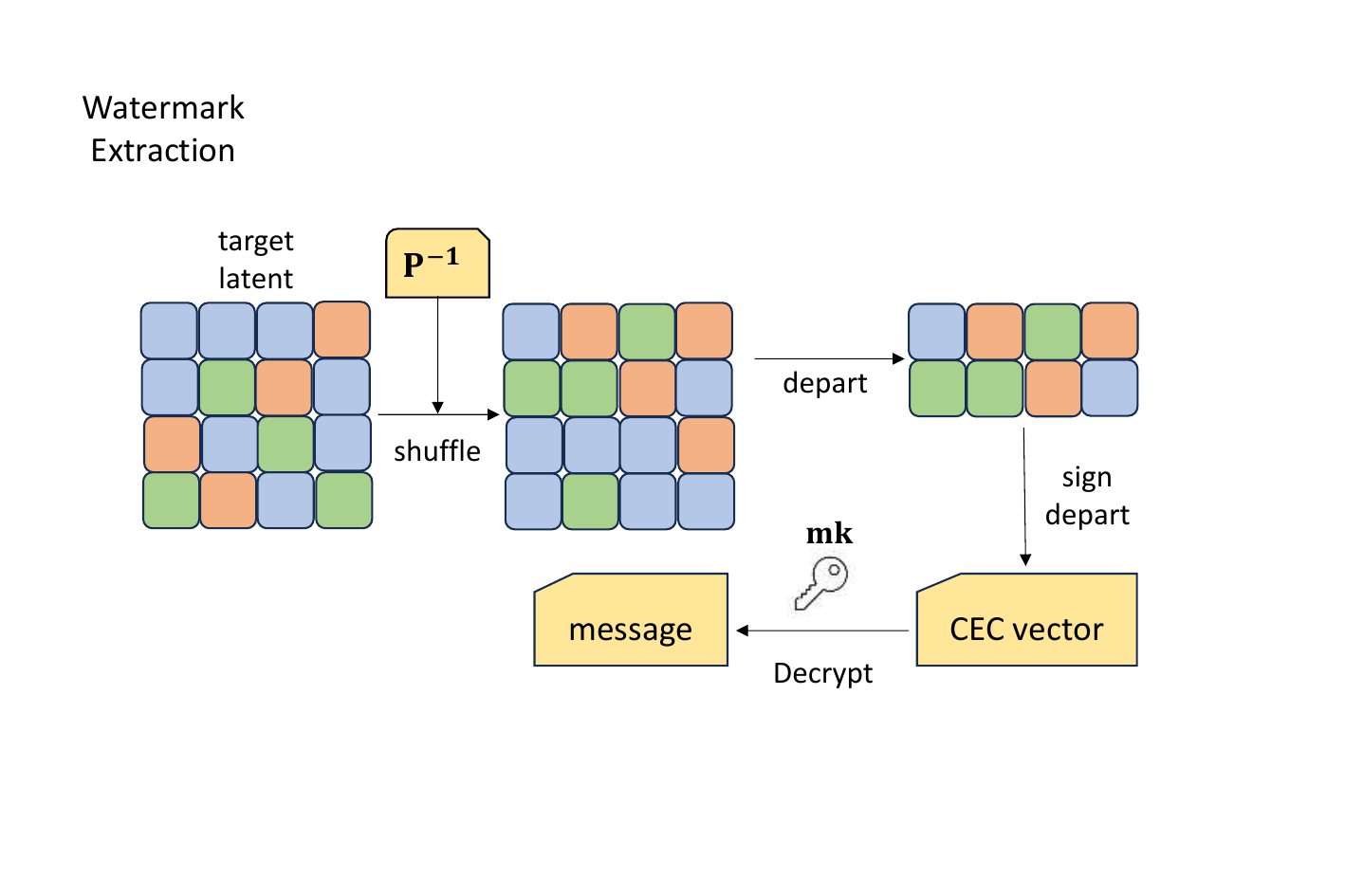}
  \caption{Watermark extraction of CECMark. After inverting the shuffle, the watermark bits are directly extracted as the signs of the first $n$ positions. }
  \label{fig:watermark_extraction}
\end{figure}

\section{Example Images of Different Schemes}
\label{sect:example}

\begin{figure}[H]
\centering
\includegraphics[width=1\linewidth]{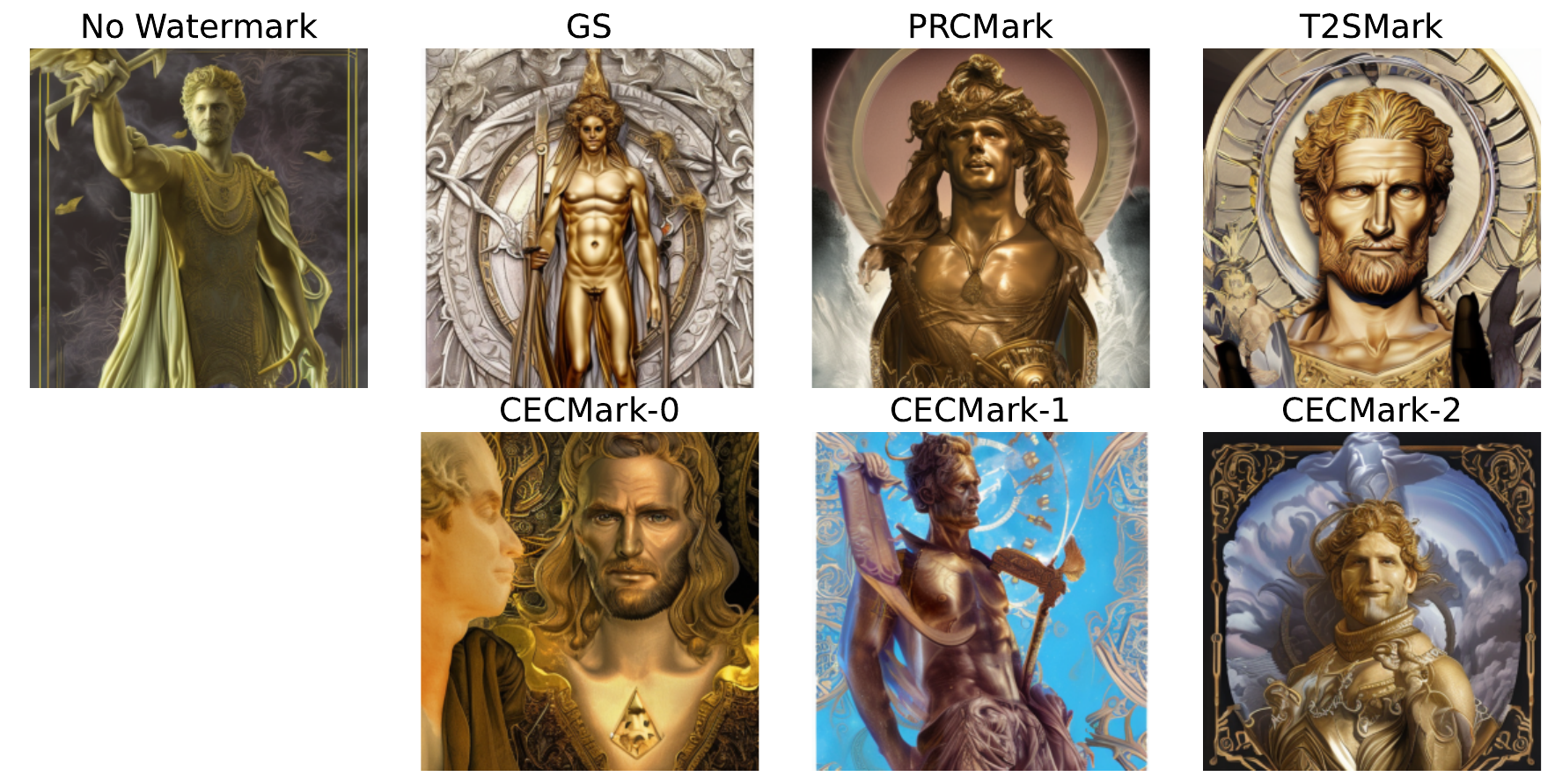}
\caption{Images generated with prompt ``\textit{ultra realistic illustration, a golden statue of a herculean glenn howerton as the god apollo, intricate, elegant, highly detailed, digital painting, artstation, concept art, smooth, sharp focus, illustration, art by artgerm and greg rutkowski and alphonse mucha}'' under different schemes.}
\label{fig:image-examples}
\end{figure}

\section{Details of Empirical Distinguishability}
\label{sect:hyper}
\begin{figure}[t]
    \centering
    \includegraphics[width=1\linewidth]{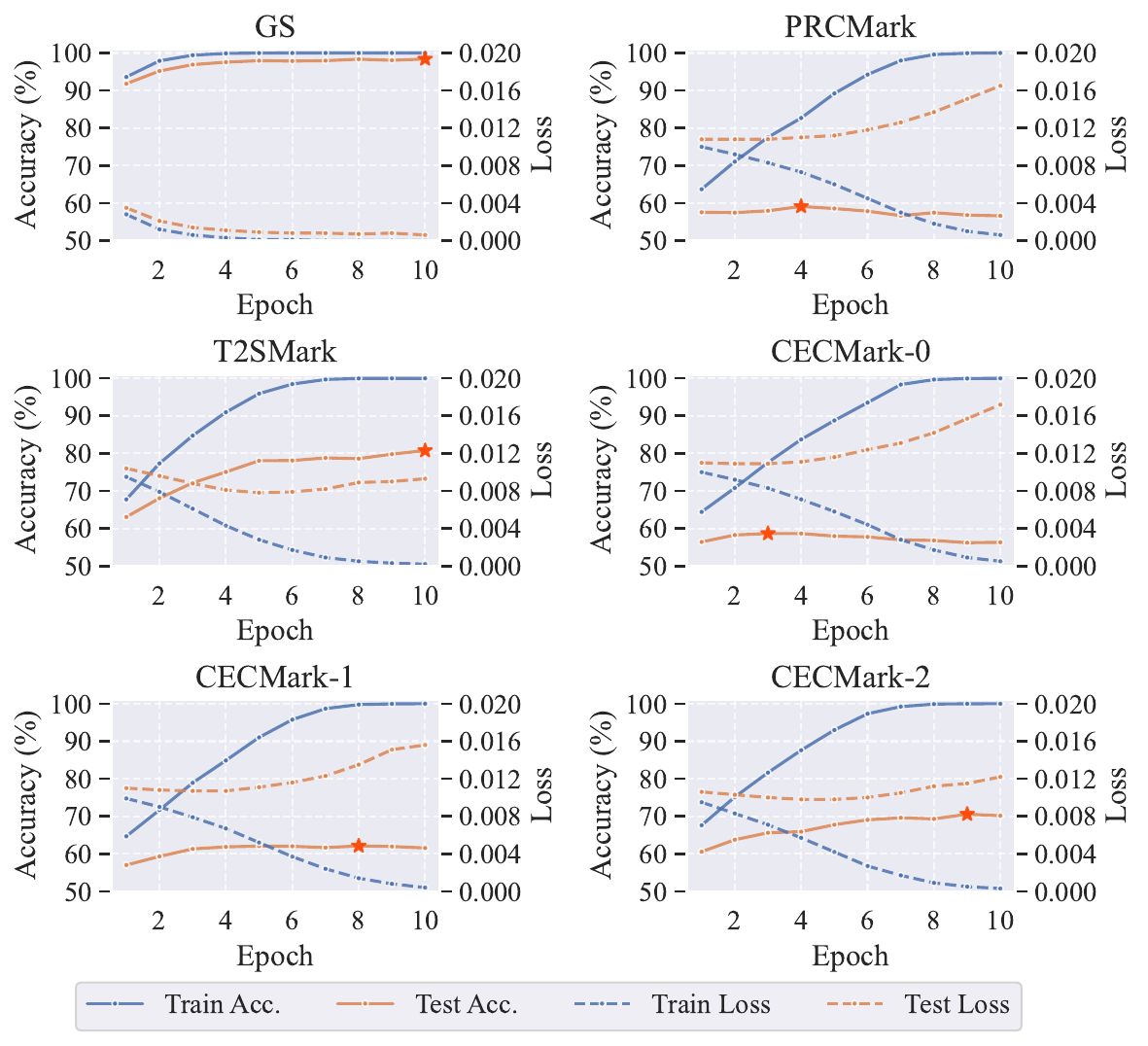}
    \caption{The loss and accuracy curves for ResNet-18 classifiers trained to distinguish watermarked from unwatermarked images.}
    \label{fig:distinguish}
\end{figure}

\mypara{Experiment Settings}
For each of the 1,000 prompts, we generate 9 watermarked and 9 unwatermarked images for training, and 1 watermarked and 1 unwatermarked image for testing, all using the same key.
This results in 18,000 training samples and 2,000 testing samples.
Each model is trained from the pre-trained ResNet-18 model for 10 epochs with a batch size of 64, using the Adam optimizer at a learning rate of $10^{-5}$.

\mypara{Experiment Result}
The loss and accuracy on the training and testing datasets for each scheme are shown in~\Cref{fig:distinguish}.
GS watermarked images can be distinguished with near-perfect accuracy, and T2SMark also reaches over 80\% distinguishability, indicating these schemes have limited confidentiality.
In contrast, CECMark-0 achieves the lowest distinction accuracy, which remains below 60\%.
CECMark-1 is slightly more distinguishable than PRCMark but still remains around 60\%.
However, CECMark-2 exhibits only moderate confidentiality, with distinguishability rising to about 70\%.
This increase in distinguishability among the CECMark variants may be attributed to the mixture of the \cecinst code and true randomness, which introduces a larger distributional shift and makes the watermark more distinguishable.

\section{Instantiability of PRC}
\label{sect:inst_prc}
In this part, we present a comprehensive analysis of the limitations of all currently known PRC instances.
We divide the current PRCs into two categories: secret-key PRCs and public-key PRCs.

\subsection{Secret-key PRC}
The first category is the secret-key version, including PRF-based PRCs~\cite{moitra2024edit,ghentiyala2024new} and spaced-bounded PRCs in~\cite{ghentiyala2024new}.
\begin{itemize}
\item \textbf{Local PRFs-Based PRC~\cite{moitra2024edit}.}  The core idea of this instance is to embed a pattern that links the input and output of a PRF, specifically $(\vec{x}\,\|\,F(\vec{x})+\vec{e})$, where $F:\{0,1\}^{n(\lambda)} \rightarrow \{0,1\}$ is a local PRF depending only on $\Omega(\log(n))$ bits of the input, and $\vec{e}$ represents a Bernoulli error.
However, this locality introduces a significant vulnerability to statistical analysis. An adversary against undetectability can collect $O(n^2)$ samples $\{\vec{y}^{(j)}\}$ and examine the joint distribution of coordinate pairs $(\vec{y}^{(j)}_{\alpha}, \vec{y}^{(j)}_{\beta})$ for all indices $\alpha, \beta$.
While these pairs are statistically independent in a truly uniform distribution, the watermarked distribution exhibits a detectable bias. Specifically, if $\vec{y}_{\alpha}$ corresponds to an input bit of $F$ and $\vec{y}_{\beta}$ to its output, the conditional probability satisfies:
\[
\Pr[\vec{y}_{\beta}^{(j)} = \vec{y}_{\beta}^{(k)} \mid \vec{y}_{\alpha}^{(j)} = \vec{y}_{\alpha}^{(k)}] 
\ge \frac{1}{2} + \frac{1}{2^{\Omega(\log n)}} 
= \frac{1}{2} + \frac{1}{\Omega(n)}
\]
where the bias $\frac{1}{\Omega(n)}$ arises whenever the $\Omega(\log n)$ bits influencing $F$ happen to coincide. This indicates that while local PRFs offer efficiency, they fail to provide robust undetectability against polynomial-time distinguishers.

\item \textbf{PRF-Based PRC~\cite{ghentiyala2024new}.} A different approach is taken in~\cite{ghentiyala2024new}, which proposes a zero-bit construction utilizing a PRF with a larger input domain of $O(n)$ instead of $\Omega(\log n)$ to mitigate such local correlations.
Although this design achieves stronger undetectability, the scheme can only tolerate \textit{sub-constant error rates}, rendering it impractical for noisy real-world channels.

\item \textbf{Space-bounded PRC~\cite{ghentiyala2024new}.}
In Section 4 of~\cite{ghentiyala2024new}, the authors propose a space-bounded PRC of the form $(\vec{a}\,\|\, \vec{a} \cdot \vec{s}_1+\vec{e}_1 \| \cdots \| \vec{a} \cdot \vec{s}_k+\vec{e}_k)$, where $\{\vec{s}_i\}$ are sparse vectors and $\{\vec{e}_i\}$ are Bernoulli errors.
Due to the sparsity of $\{\vec{s}_i\}$, this construction remains vulnerable to statistical analysis.
Accordingly, the authors only claim undetectability against adversaries with $O(n^{1.5})$ space complexity, which limits its practical security guarantees.

\end{itemize}

\subsection{Public-key PRC}
The second category comprises public-key constructions that typically rely on hardness assumptions~\cite{christ2024pseudorandom,DBLP:conf/iclr/GunnZS25,ghentiyala2024new,christ2025improved}.

\begin{itemize}
    \item \textbf{Planted XOR-Based PRC~\cite{christ2024pseudorandom,DBLP:conf/iclr/GunnZS25,ghentiyala2024new}.}
    Schemes in~\cite{christ2024pseudorandom,DBLP:conf/iclr/GunnZS25} and Section~5 of~\cite{ghentiyala2024new} are based on the planted XOR problem.
    These constructions are known to be vulnerable to quasi-polynomial-time attacks.

    \item \textbf{Planted Hyper-loop-Based PRC~\cite{ghentiyala2024new}.}
    The construction in Section~4 of~\cite{ghentiyala2024new} relies on the planted hyper-loop problem, which is also susceptible to quasi-polynomial-time attacks.

    \item \textbf{Reed-Solomon-Based PRC~\cite{christ2025improved}.}
    Recent attempts in~\cite{christ2025improved} employ Reed-Solomon codes to construct PRCs, but exhibit significant theoretical and practical shortcomings. Specifically, the construction in Theorem 5.6 is limited to deletion channels, while the construction in Theorem 5.7 requires a decoder with exponential-time complexity, making it computationally infeasible. 
    As a result, these constructions do not fully meet the PRC requirements in their current form.
\end{itemize}


Crucially, most of these works lack concrete parameter instantiations, remaining purely theoretical. 
Even for the scheme in~\cite{DBLP:conf/iclr/GunnZS25}, which does provide parameters, recent cryptanalysis in~\cite{wang2025cryptanalysis} demonstrates that its undetectability can be compromised by a polynomial-time adversary. 
Consequently, existing PRCs remain far from practical deployment.

\section{Attacks against Confidentiality of PRCMark}
\label{sect:prc_attack}
We further adopt Attack-I in~\cite{wang2025cryptanalysis} to break the confidentiality of PRCMark. 
Specifically, we first recover $l$ rows of the secret key $\mat{P}$ using Attack-I, and construct a partial secret key $\mat{P}' \in \GF{2}^{l \times n}$. 
Then, we apply the $\Decode()$ algorithm exactly as in~\cite{DBLP:conf/iclr/GunnZS25} to decode the message.

To verify the effectiveness of our attack, we apply PRCMark with $t=3$ on stable-diffusion-2-1-base~\cite{Rombach_2022_CVPR}. 
We randomly generate $100$ pairs of $\vec{pk}=(\mat{G},\vec{z})$. 
For each $\vec{pk}$, we generate $100$ images. Given each public key $\vec{pk}$, we recover a partial secret key $\mat{P}'$ by setting $l=10{,}000$. 
The results show that we successfully recover $97.99\%$ of the messages, indicating that the confidentiality of PRCMark is not satisfied.

\end{document}